\documentclass[sigconf, nonacm, pdfa]{acmart}

\newcommand\vldbdoi{XX.XX/XXX.XX}
\newcommand\vldbpages{XXX-XXX}
\newcommand\vldbvolume{17}
\newcommand\vldbissue{11}
\newcommand\vldbyear{2024}
\newcommand\vldbauthors{\authors}
\newcommand\vldbtitle{\shorttitle} 
\newcommand\vldbavailabilityurl{https://github.com/Emory-AIMS/priv_traj_gen}
\newcommand\vldbpagestyle{empty} 

\newcommand{\partitle}[1]{\smallskip \noindent \textbf{#1.}}

\usepackage{tcolorbox}
\usepackage{graphicx}
\usepackage{subcaption}
\usepackage[noend]{algpseudocode}
\usepackage{algorithm}

\usepackage{multirow}
\usepackage{ulem}
\usepackage[a-2b]{pdfx}
\usepackage{balance}

\begin{document}

\title{HRNet: Differentially Private Hierarchical and Multi-Resolution Network for Human Mobility Data Synthesization}

\author{Shun Takagi *}\thanks{* Work partially done while visiting Emory University}
\affiliation{%
  \institution{Kyoto University}
}
\email{shun021677@gmail.com}

\author{Li Xiong}
\affiliation{%
  \institution{Emory University}
}
\email{lxiong@emory.edu}

\author{Fumiyuki Kato}
\affiliation{%
  \institution{Kyoto University}
}
\email{fumilemon79@gmail.com}

\author{Yang Cao}
\affiliation{%
  \institution{Tokyo Institute of Technology}
}
\email{cao@c.titech.ac.jp}

\author{Masatoshi Yoshikawa}
\affiliation{%
  \institution{Osaka Seikei University}
}
\email{yoshikawa-mas@osaka-seikei.ac.jp}


\begin{abstract}
Human mobility data offers valuable insights for many applications such as urban planning and pandemic response, but its use also raises privacy concerns. In this paper, we introduce the Hierarchical and Multi-Resolution Network (HRNet), a novel deep generative model specifically designed to synthesize realistic human mobility data while guaranteeing differential privacy. 
We first identify the key difficulties inherent in learning human mobility data under differential privacy. In response to these challenges, HRNet integrates three components: a hierarchical location encoding mechanism, multi-task learning across multiple resolutions, and private pre-training. 
These elements collectively enhance the model's ability under the constraints of differential privacy. 
Through extensive comparative experiments utilizing a real-world dataset, HRNet demonstrates a marked improvement over existing methods in balancing the utility-privacy trade-off.
\end{abstract}

\maketitle

\pagestyle{\vldbpagestyle}
\begingroup\small\noindent\raggedright\textbf{PVLDB Reference Format:}\\
\vldbauthors. \vldbtitle. PVLDB, \vldbvolume(\vldbissue): \vldbpages, \vldbyear.\\
\href{https://doi.org/\vldbdoi}{doi:\vldbdoi}
\endgroup
\begingroup
\renewcommand\thefootnote{}\footnote{\noindent
This work is licensed under the Creative Commons BY-NC-ND 4.0 International License. Visit \url{https://creativecommons.org/licenses/by-nc-nd/4.0/} to view a copy of this license. For any use beyond those covered by this license, obtain permission by emailing \href{mailto:info@vldb.org}{info@vldb.org}. Copyright is held by the owner/author(s). Publication rights licensed to the VLDB Endowment. \\
\raggedright Proceedings of the VLDB Endowment, Vol. \vldbvolume, No. \vldbissue\ %
ISSN 2150-8097. \\
\href{https://doi.org/\vldbdoi}{doi:\vldbdoi} \\
}\addtocounter{footnote}{-1}\endgroup

\ifdefempty{\vldbavailabilityurl}{}{
\vspace{.3cm}
\begingroup\small\noindent\raggedright\textbf{PVLDB Artifact Availability:}\\
The source code, data, and/or other artifacts have been made available at \url{\vldbavailabilityurl}.
\endgroup
}

\section{Introduction}
\label{sec:introduction}
In recent years, the use of human mobility data has gained significant attention for its potential to contribute to societal benefits, such as traffic forecasting, urban planning, and pandemic response, including COVID-19 spread analysis \cite{wang2020deep, chang2021mobility, lin2019deepstn+}.
However, it also raises critical privacy concerns even if the data are aggregated and anonymized~\cite{tu2018new, pellungrini2018analyzing}.

Differential Privacy (DP)~\cite{dwork2006differential} has emerged as the leading standard for maintaining data privacy. 
DP offers a strong, mathematically grounded privacy guarantee without relying on restrictive assumptions about potential adversaries. 
The core principle of DP is to ensure that output from the data analysis does not substantially differ (bounded by a privacy parameter or privacy budget), regardless of whether any specific individual's data is included or excluded from the dataset. 
The broad applicability and growing adoption of DP in the real world~\cite{erlingsson2014rappor,apple,ding2017collecting,johnson2018towards}, as well as its endorsement by the US Census Bureau~\cite{bureaumap}, signify its importance and effectiveness.

\begin{figure}[t]
    \centering
    \includegraphics[width=.9\linewidth]{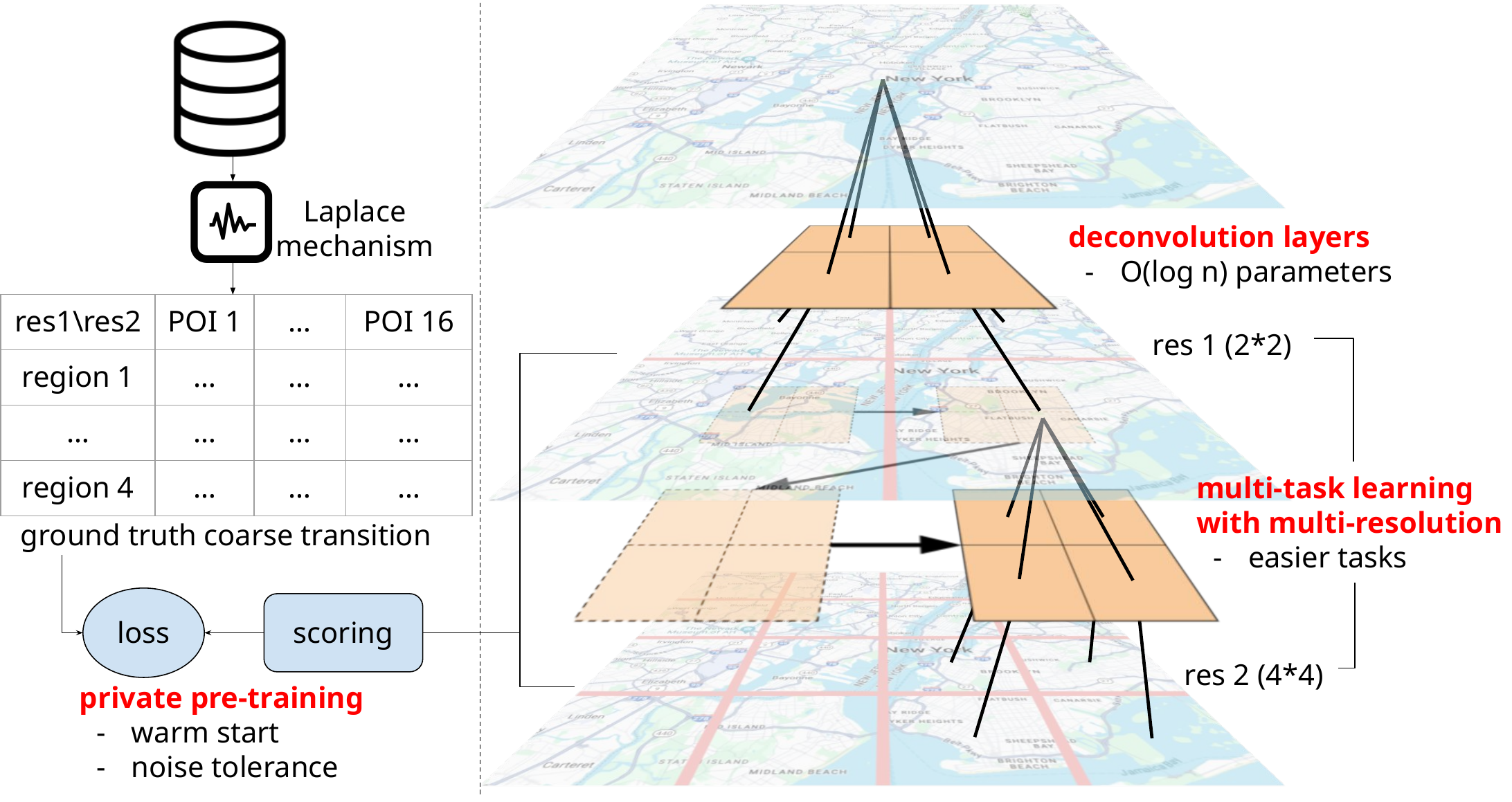}
    \caption{HRNet utilizes three novel components to address the bottlenecks of applying DP-SGD in human mobility generation: 1) a hierarchical location encoding mechanism using deconvolution networks, 2) multi-task learning across multiple resolutions, and 3) private pre-training using a DP coarse transition matrix. 
    }
    \label{fig:overview}
\end{figure}

In this paper, we study the problem of DP data synthesis~\cite{zhang2017privbayes, bowen2020comparative, wang2021dpsyn} for human mobility data.
DP data synthesis generates synthetic data that possesses statistical properties similar to the real data while ensuring DP. 
According to the post-processing theorem of DP, unlimited analysis can be performed on this synthetic data without introducing further privacy concerns.
Our goal is to synthesize human mobility data, defined as a sequence of locations, to closely resemble the real human mobility data, which can be then used for a variety of downstream tasks in previously mentioned applications.

Existing DP data synthesis methods for human mobility primarily rely on DP-aware data structures, which can be categorized into tree-based \cite{chen2012differentially, he2015dpt, cai2021trajectory}, Markov-based \cite{wang2023privtrace, gursoy2018differentially}, and clustering approaches \cite{liu2021differentially, hua2015differentially, li2017achieving}. 
These methods generally employ generalization to reduce dimensionality and sensitivity, thereby balancing utility and privacy, but often at the cost of information loss.
For instance, clustering approaches only generate transitions between $k$ centroids by location generalization using clustering, leading to a loss of precise transition information within the clusters.
Moreover, the state-of-the-art Markov-based method~\cite{wang2023privtrace} does not account for correlations beyond two-step intervals, and many methods overlook auxiliary information such as time, which is crucial in human movement.
Tree-based methods aim to minimize such loss by constructing a data-dependent tree with a portion of the privacy budget. 
However, such consumption leads to a lack of privacy budget during the data generation phase, which results in diminished utility.

MTNet~\cite{wang2022deep} is a deep learning approach designed for generating short trips on road networks, such as taxi trips, while ensuring DP. 
By utilizing meta-information from road networks, which is not sensitive, MTNet adeptly manages the utility-privacy trade-off.
However, this approach imposes the inherent limitation of restricting mobility to road networks.
Human mobility, in reality, is not confined to such networks; it encompasses various modes including subway travel and movement through non-road spaces like parks. 
Additionally, even when focused on road networks, human mobility often evolves into lengthy sequences of road segments, a complexity that poses significant challenges under DP constraints. Consequently, adapting MTNet for a broader spectrum of human mobility appears to be an unpromising approach.

To the best of our knowledge, this is the first paper to adopt a deep learning approach with DP for generating human mobility which is not constrained to road networks.
Recent advancements have shown that deep learning approaches handling human mobility data outperform traditional models~\cite{wang2020deep, LucaBLP23-Survey-Mobility}. There are several generative models such as adversarial network (GAN) based models for generating realistic synthetic trajectories \cite{zhang2023csgan, feng2020learning}.
However, these models do not ensure formal privacy guarantee.
Theoretically, it is possible to ensure DP for any deep learning models by replacing the optimization method with differentially private stochastic gradient descent (DP-SGD)~\cite{abadi2016deep}.
Unfortunately, simply adopting DP-SGD on these generative models results in either poor utility or infeasibly high privacy cost due to the large space domain and complexity of the models.
More concretely, we elaborate the two primary challenges as follows. 
\begin{enumerate}
    \item  \textbf{The number of model parameters:} It has been demonstrated~\cite{bassily2014private, bassily2020stability} that DP-SGD inevitably increases the lower bound of empirical risk as the number of parameters increases. In most models about human mobility (e.g., RNN model~\cite{yang2020location} and attention model~\cite{yan2023spatio}), the number of model parameters depends on the size of the space domain. If we encode the space or Points of Interests (POIs) using an encoding scheme such as embedding matrices or linear embeddings, the number of model parameters is $O(n_{\rm POI})$  where $n_{\rm POI}$ denotes number of POIs, which can be significantly large. 
    \item \textbf{Learning difficulty:} the difficulty of learning human mobility escalates with  increasing $n_{\rm POI}$ due to the increasing number of labels for training. 
    Therefore, a larger $n_{\rm POI}$ necessitates larger models and longer training epochs to discern subtle POI differences.
    This leads to decreased utility due to the large number of parameters and higher consumption of privacy budgets due to the composition theorem~\cite{kairouz2015composition}.
\end{enumerate}

\partitle{Contributions}
In response to these challenges, we present the Hierarchical and Multi-Resolution Network (HRNet), which encompasses a novel network structure and learning methodology. 
The key features of HRNet, as illustrated in Figure~\ref{fig:overview}, are threefold:
\begin{enumerate}
    \item \textbf{Hierarchical location embedding:} to alleviate the first issue, HRNet adopts a hierarchical structure for location embedding via transposed convolution. This approach significantly reduces the number of parameters to $O(\log (n_{\rm POI}))$, thereby alleviating the decrease in utility commonly associated with traditional embedding methods that have $O(n_{\rm POI})$ parameters.
    
    \item \textbf{Multi-task learning with multiple resolutions:} to alleviate the second issue, we introduce multi-task learning. Our hierarchical network design enables encoding of multiple resolutions. 
    Beyond mastering the primary task, the network concurrently engages in learning additional, less complex tasks with coarser resolutions. 
    This approach allows for the intricate primary task to be deduced from these simpler sub-tasks, thereby efficiently mitigating the learning difficulties associated with a large number of POIs.
    
    \item \textbf{Effective and private pre-training:} to alleviate both issues, we conduct private pre-training.
    The effectiveness of pre-training in addressing the first issue of DP-SGD has been underscored~\cite{amid2022public, kairouz2021nearly}.
    Additionally, pre-training provides a `warm start', effectively reducing the need for extensive training epochs, thus alleviating the second issue.
    A common challenge with pre-training is its reliance on public data, which may not always be readily available or suitable for all scenarios.
    HRNet enables efficient pre-training without public data by leveraging a dense DP-compatible transition matrix.
\end{enumerate}


\section{Preliminaries}

\subsection{Differential Privacy}
\label{sec:differential-privacy}
Differential Privacy (DP), as introduced by Dwork et al. \cite{dwork2006differential}, provides a robust mathematical framework for quantifying privacy leakage in data publication scenarios. 
The formal definition of DP is as follows:
\begin{definition}[($\varepsilon, \delta$)-Differential Privacy]
Consider a dataset domain $\mathcal{D}$ and output domain $\mathcal{Z}$. Given $\varepsilon\in \mathbb{R}^+$ and $\delta\in [0,1]$, a randomized mechanism $m$, which randomly outputs $m(D)\in \mathcal{Z}$ with input $D\in\mathcal{D}$, satisfies ($\varepsilon, \delta$)-DP if, for any two datasets $D, D' \in \mathcal{D}$ differing in at most one element, and for any subset of outputs $Z \subseteq \mathcal{Z}$, the following inequality holds:
\begin{equation}
\nonumber
\Pr[m(D)\in Z] \leq e^\varepsilon \Pr[m(D')\in Z] + \delta.
\end{equation}
\end{definition}

\textbf{DP-SGD:}
Differentially Private Stochastic Gradient Descent (DP-SGD), as developed by Abadi et al. \cite{abadi2016deep}, adapts the conventional SGD optimization process for deep learning models to satisfy DP. 
In traditional SGD, parameters of a model $\theta$ are iteratively updated to minimize empirical risk, with gradient computations performed using sampled data (minibatches). 
DP-SGD modifies this process to ensure DP by introducing gradient clipping and noise addition steps. 
Specifically, the $l_2$ norm of the gradient is first clipped to limit sensitivity, followed by the addition of the Gaussian noise to the averaged clipped gradients. 
The model parameters are then updated using these randomized gradients, similar to conventional SGD.
This process continues until convergence is achieved or the allocated privacy budget is depleted. 
\textcolor{black}{
From a DP perspective, this process is interpreted as the sequential composition of the Gaussian mechanism with subsampling (i.e., minibatch). Therefore, by employing composition techniques, the final published model satisfies DP. 
We adopt numerical composition, which numerically computes the upper bound of the privacy loss parameter $\varepsilon$ using privacy loss distributions~\cite{sommer2018privacy}. 
Several studies enhance its computation using techniques such as the Fast Fourier Transform~\cite{koskela2020computing} and the "connect-the-dots" algorithm~\cite{dockhorn2022differentially}. 
For implementation, we utilize Google’s differential-privacy library\footnote{\url{https://github.com/google/differential-privacy} (Accessed: July 14, 2024)}.
}

\subsection{Problem Formulation}
\label{sec:problem formulation}
This study aims to generate synthetic human mobility data that accurately mirrors actual human mobility (e.g., daily human mobility such as home to workplace to home). 
Due to the inherently continuous nature of human mobility in terms of both time and geographical coordinates (latitude and longitude), this presents significant challenges in direct modeling and evaluation.
To address this, we adopt a conventional approach of discretizing human mobility data, thereby simplifying the modeling process. 
This section first formulates discretized human mobility representation. 
It then details the formulation of a generator for discretized human mobility and outlines the evaluation methodology.

\subsubsection{Trajectory:}
\label{sec:staypoint}
As a fundamental step, we define Points of Interests (POIs)\footnote{Here, the term `POI' is not used in its original sense of indicating semantic locations. However, our method can be applied in its original context as described in Section~\ref{sec:loc_assign}.} as uniformly distributed grid cells on a map, with each grid cell assigned a unique integer ID. 
For instance, if a target map is divided into a $w \times w$ grid, the domain of POIs is defined as $L := [w^2]$, where $[n]$ denotes the set $(1,2,\dots,n)$. 
That is, $n_{\rm POI}=|L|$.
As illustrated in Figure~\ref{fig:overview}, different levels of discretization are possible, such as $1 \times 1$, $2 \times 2$, and $4 \times 4$.
A geographical location, defined by its latitude and longitude, is transformed into an integer representing the grid cell that encompasses the location.
The mapping is performed such that the grid cell located at $(\lfloor l/w\rfloor,l\% w)$ is assigned the integer $l\in L$.
We then represent a human mobility sequence as a sequence of stay points, denoted as
$
\mathbf{v}=[v_1,v_2,\dots],
$
where $i\in \mathbb{N}, v_i = (l_i, t_i)$, $l_i\in L$.
Note that this representation aligns POIs that change from the previous POI so that $l_i \neq l_{i+1}$.
Here, $t_i$ is the index of the time slot that includes the corresponding time.
Time slots are made by discretizing entire time by $n_{\rm time}$, where $n_{\rm time}$ is the number of time slots.
For example, $v_2$ indicates the tuple of the second POI visited and the time at which this POI was reached. 
From this point forward, the term `trajectory' will refer to this discretized representation of human mobility.

\subsubsection{Generator:}
In our study, we employ a neural network-based generator, denoted as $G_\theta$, to stochastically generate synthetic trajectories.
The process can be mathematically represented as
$
\mathbf{v} \sim G_\theta,
$
where $\mathbf{v}$ represents the generated synthetic trajectory, and $\theta$ symbolizes the trainable parameters of the neural network.
A detailed example of such a generator, including its architecture and operational mechanics, will be discussed in Section~\ref{sec:baseline}.

\subsubsection{Evaluation:}
Direct evaluation of the performance of $G_\theta$ is intractable due to the inherent sparsity of trajectory data.
To overcome this, we adopt an indirect evaluation approach, focusing on various statistical properties of trajectories. 
These statistical properties are computed from the dataset generated by $G_\theta$ and are then compared with empirical statistics derived from the provided trajectory dataset $D$.
Key aspects of trajectory data that we examine include waypoints, routes, destinations, transitions, and travel distances. 
By analyzing these specific statistics, we can gain insights into the accuracy and fidelity of the trajectories generated by $G_\theta$.
The formal definitions and methodologies for computing these statistics are  detailed in Section~\ref{sec:exp_setup}. 
This approach enables a thorough and nuanced evaluation of the generator's performance in replicating realistic human mobility patterns.

\subsection{The Baseline Generator}
\label{sec:baseline}
This section delineates the architecture and operational mechanics of the baseline generator used in our study. 

\subsubsection{The baseline generator}
The baseline generator is structured upon the chain rule of conditional probabilities:
\begin{equation}
\label{eq:chain}
G_\theta=\Pr(\mathbf{v}) = \prod_{i=1}^{\mathbf {|v|}} \Pr(v_i|v_1,...,v_{i-1}).
\end{equation}
This formulation allows the generator to sequentially construct a trajectory from $i=1$ by modeling each conditional probability $\Pr(v_i|v_1,\dots,v_{i-1})$.
To accomplish this, the generator integrates three core components: embedding matrices for data points $v_i$, a recurrent neural network (RNN) to process sequences, and a scoring component to make a distribution.

\textbf{Embedding Matrix:}
The embedding matrix translates each data point $v_i= (l_i, t_i)$ to a vector representation.
Two separate embedding matrices are employed: one for the POI ($l_i$) and another for time ($t_i$). 
The POI embedding matrix $M_{\rm POI}$ contains $|L|$ trainable vectors, whereas the time embedding matrix $M_{\rm time}$ includes $n_{\rm time}$ trainable vectors.
The vector representation of $v_i$ is formed by concatenating the $t_i$th vector from the time embedding matrix with the $l_i$th vector from the POI embedding matrix:
\begin{equation}
\label{eq:baseline_encode}
{\rm encode}(v_i)=[M_{\rm time}[t_i], M_{\rm POI}[l_i]].
\end{equation}
Here, $M[i]$ represents the access to the vector of $M$ at $i$th index and $[a,b]$ represents the concatenation of vectors $a$ and $b$.

\textbf{Recurrent Neural Network:}
For encoding trajectory prefixes $(v_1, \dots , v_{i-1})$, we employ Gated Recurrent Units (GRUs)~\cite{cho-etal-2014-learning}.
This is because we have empirically found that GRU is slightly more effective than alternative methods such as LSTM~\cite{hochreiter1997long} and attention mechanisms~\cite{vaswani2017attention} in the context of DP-SGD.
We attribute this to the GRU's simpler architecture and reduced parameters.
The GRU cell updates its state based on the previous state $h_{i-1}$ and the encoded $v_{i-1}$:
$
h_i = f_{\rm GRU}(h_{i-1}, {\rm encode}(v_{i-1})).
$
This process is applied recursively to encode the entire trajectory prefix.

\textbf{Scoring component:}
The scoring component is responsible for converting the encoded prefix into the probability distribution over $[n_{\rm POI}]\times [n_{\rm time}]$ for next predicted location.
This conversion is accomplished using feed-forward neural networks $g_{\rm POI}:\mathbb{R}^{n_{\rm hidden}}\to\mathbb{R}^{|L|}$ and $g_{\rm time}:\mathbb{R}^{n_{\rm hidden}}\to\mathbb{R}^{n_{\rm time}}$ where $n_{\rm hidden}$ is the dimensionality of the hidden state $h_i$, followed by a softmax function to ensure proper probabilistic normalization. 
The probability distribution over spatial location $l\in L$ given the hidden state $h_i$ is computed as
$\Pr(l=l_i|h_i)={\rm softmax}(g_{\rm POI}(h_i))[l_i].$
Similarly, the probability distribution over time $t\in [n_{\rm time}]$ given $h_i$ is determined by
$\Pr(t=t_i|h_i)={\rm softmax}(g_{\rm time}(h_i))[t_i].$
Ultimately, the conditional probability $\Pr(v_i|v_1,\dots,v_{i-1})$ given the previous observations is the product of these two distributions, representing a joint probability over POI and time:
$
\Pr(v_i|v_1,\dots,v_{i-1})=\Pr(l=l_i|h_i)\Pr(t=t_i|h_i).
$

\subsubsection{The objective function:}
\label{sec:objective function}
To optimize the baseline generator, given a training trajectory sequence $\mathbf{v}=(v_1,\dots,v_i)$, we employ cross entropy loss as our objective function. 
The loss is calculated as follows:
\begin{align}
&\nonumber{\rm loss}_{\rm POI}+{\rm loss}_{\rm time}=\\
&\nonumber\sum_{i=1}^{|v|}\left(\sum_{c\in L} \delta_{c,l_i}\log (\Pr(l=l_i|h_i))+\sum_{c\in [n_{\rm time}]} \delta_{c,t_i}\log (\Pr(t=t_i|h_i))\right),
\end{align}
where $\delta_{i,j}$ is the Kronecker delta.

\subsubsection{The sequential generation:}
The trained generator is capable of generating human mobility sequences by employing the chain rule as delineated in Equation~(\ref{eq:chain}). 
Specifically, the process begins with the initial POI.
Subsequently, each following POI is recursively sampled based on the conditional probability that is informed by the previously sampled POIs.
This procedure is iteratively executed until a predefined condition, such as reaching a maximum sequence length or a specific vocabulary signifies the end of the sequence.

\subsection{The two bottlenecks}
\label{sec:bottleneck}
The baseline model discussed earlier shows a decrease in performance when DP-SGD is applied, particularly as the number of POIs increases. 
In this section, we identify and elaborate on the two primary factors causing these performance bottlenecks.

\subsubsection{The number of parameters $|\theta|$}
Under the constraints of DP, it is known that the lower bound of empirical risk of the private model polynomially scales~\cite{bassily2014private, bassily2020stability} as the number of model parameters  increases.
In the baseline model, the size of the parameters  $\theta$ expands linearly with  increasing number of POIs ($n_{\rm POI}$) due to the location embedding matrix and the scoring component.
Consequently, the empirical risk would polynomially escalate as $n_{\rm POI}$ increases, which leads to worse utility of the generated human mobility.

\subsubsection{The number of labels}
With an increase in the number of class labels (i.e., $n_{\rm POI}$), the model faces challenges in discerning more subtle feature distinctions between classes (see Section~\ref{sec:objective function}). 
This often necessitates the use of larger model architectures or extended training epochs. 
However, such strategies are impractical under DP-SGD constraints. 
In addition to the increased empirical risk due to the increased model size as mentioned above, increasing the number of training epochs results in greater privacy loss due to the composition theorem of DP~\cite{kairouz2015composition}.
Therefore, learning difficulty caused by the large number of labels is a bottleneck of DP-SGD.

\textcolor{black}{
\subsubsection{Discussion: impact of the discretization parameter}
\label{sec:discussion_about_discretization}
In many deep learning methods~\cite{wang2023privtrace, lian2020geography, lim2022hierarchical}, as in this study, latitudes and longitudes are discretized using a grid. However, determining the parameters for this discretization (i.e., $w$ in this study) is not straightforward. 
This is because finer discretization makes the data more sparse and increases the number of parameters required for the network (location encoding).
As discussed above, DP-SGD leads to a decrease in accuracy in these cases. 
However, finer grids have the advantage of capturing more detailed movements between grid cells, improving the quality of the training data and potentially increasing accuracy.
Thus, there is a trade-off between the accuracy improvement from higher-quality training data and the accuracy decrease due to the noise required by DP.
Therefore, mitigating the two bottlenecks makes it less susceptible to the negative effects of high-resolution discretization and thereby improves accuracy.
}


\subsection{Transposed Convolution}
\label{sec:deconvolution}
In this section, we delve into the transposed convolutional operation, a key element of the main component in our proposed model. 
Originating from the domain of computer vision~\cite{long2015fully}, the transposed convolution layer, also known as a deconvolution layer, is typically employed to upsample or reconstruct the original images from feature maps generated by convolutional layers. 
In contrast to its conventional usage, our model uniquely incorporates the specific instance of this layer, which we have termed the `2D quad deconvolutional network'. 
This incorporation of the transposed convolution layer is instrumental in addressing the two primary bottlenecks identified earlier.

\subsubsection*{2D quad deconvolutional network}
Consider the matrix $M$ with shape $(w, w, n_{\rm dim})$, where $n_{\rm dim}$ denotes the number of channels, and each channel comprises a $w\times w$ square matrix.
Similarly, let the deconvolutional kernel be defined as a matrix of shape $(2,2,n_{\rm dim})$, containing $n_{\rm dim}$ kernels, and each of these kernels is $2\times 2$ matrix denoted by ${\rm kernel}_{k}$.
The 2D quad deconvolutional network expands the matrix $M$ to the matrix $M^\prime$ with dimensions $(2w,2w,n_{\rm dim})$ using a deconvolutional kernel. 
See Figure~\ref{fig:overview} for the running example. 
In the second layer, the kernel slides over the input of $(2\times 2)$ matrix to perform the deconvolution operation for the four grids, expanding a $2 \times 2$ matrix into a $4 \times 4$ matrix.

Let $x$ and $y$ be integers, the transformation can be mathematically represented as:
\begin{align}
\nonumber M^\prime[2x + \delta_x, 2y + \delta_y, k] = \sum ^{n_{\rm dim}}_{k^\prime=1}{\rm kernel}_{k}[\delta_x, \delta_y, k^\prime]  M[x, y, k^\prime],
\end{align}
where, $k\in [n_{\rm dim}]$, $\delta_x$ and $\delta_y$ are binary values ($0$ or $1$), and notation $M[i, j, k]$ represents access to an element of the matrix $M$ with the index $(i,j,k)$. 
This process effectively quadruples each $n_{\rm dim}$-dimensional vector in the matrix. 
As a result, the original $n_{\rm dim}$-dimensional vector $M[x,y]$ is expanded into four distinct $n_{\rm dim}$-dimensional vectors: $M^\prime[2x, 2y]$, $M^\prime[2x+1, 2y]$, $M^\prime[2x, 2y+1]$, and $M^\prime[2x+1, 2y+1]$.
Applying the 2D quad deconvolutional layer $d$ times multiplies the number of $n_{\rm dim}$-dimensional vectors by $4^d$. 
When $w$ is a power of $2$, originating from a single $n_{\rm dim}$-dimensional root vector $\theta_{\rm root}$, the final transformation of the network can be represented as:
$$
M = {\rm deconv}^{\log_4 w^2}(\theta_{\rm root}),
$$
where ${\rm deconv}^a$ denotes the iterative application of the deconvolution operation $a$ times.

\begin{figure}[ht]
    \centering
    \includegraphics[width=\linewidth]{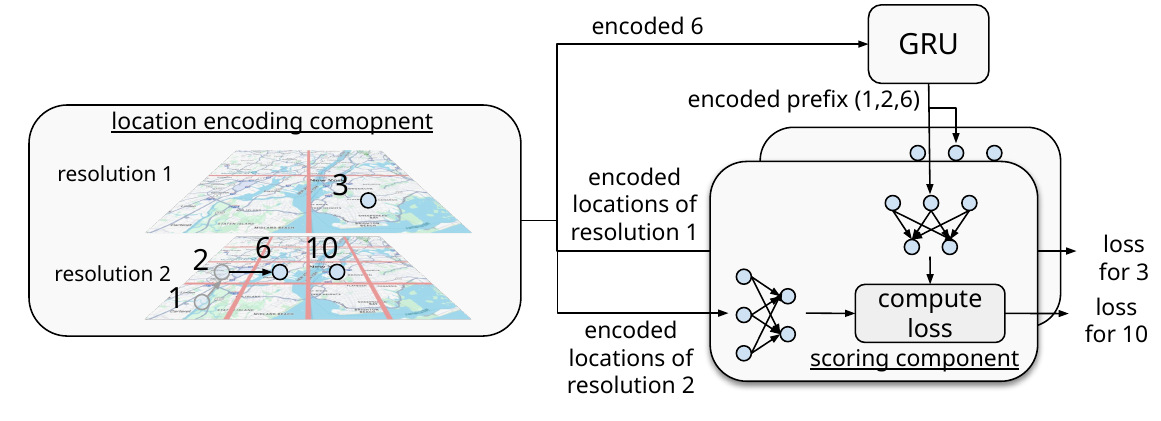}
    \caption{Overview of multi-task training with  hierarchical location encoding. In this example, we assume that $w=4$, so we have 4*4 POIs at  resolution 2 and 2*2 regions at resolution 1 due to the hierarchical location encoding. Given prefix $(1,2,6)$, the model learns to infer grid cell $10$ at resolution $2$, as well as grid cell $3$ at resolution $1$.}
    \label{fig:overview of training}
\end{figure}

\section{HRNet}
In this section, we introduce the Hierarchical and Multi-Resolution Network (HRNet). 
HRNet's core novelty is the integration of a hierarchical location encoding component, which supersedes the traditional embedding matrix and scoring component. 
This design reduces the number of parameters, effectively alleviating the first bottleneck. 
Moreover, it facilitates multi-resolution interpretation, enabling multi-task learning. 
This approach allows HRNet to process and learn from data at multiple resolutions simultaneously, adeptly handling the difficulty of learning at the finest resolution, which alleviates the second bottleneck. 
Together, they contribute to HRNet's robustness to the number of POIs $n_{\rm POI}$ when optimized with DP-SGD.

We begin by outlining HRNet, using a simple example illustrated in Figure~\ref{fig:overview of training}. 
Subsequently, we delve into detailed explanations of each component.

\subsection{Overview}
We provide the overview of multi-task training with hierarchical location encoding of HRNet, illustrated in Figure~\ref{fig:overview of training}. 
HRNet builds upon the baseline model that models conditional probabilities. 
See Section~\ref{sec:baseline} for detail of the baseline model. 
Note that we simplify our explanation, sidelining the training of time information which is the same as the baseline, to emphasize the principal distinctions.

Consider a scenario depicted in Figure~\ref{fig:overview of training}. 
The example involves two resolutions, featuring $4$ and $16$ grid cells respectively, and trajectory $(1,2,6,10)$ at resolution $2$. 
In the learning phase for $v=(1,2,6,10)$, step $1$ is learning $2$ from prefix ($1$), step $2$ is learning $6$ from prefix $(1, 2)$, and step $3$ is learning $10$ from prefix $(1,2,6)$. Consider step 3, HRNet operates as follows:

The hierarchical location encoding component encodes location $6$.
The GRU cell computes $h_3$, embedding of the prefix $(1,2,6)$, using $h_2$ and the encoded vector of $6$. 
$h_3$ forms a \textit{query vector} with the feed-forward neural networks for subsequent steps.
Unlike the baseline model that focuses exclusively on resolution $2$, HRNet concurrently considers resolution $1$, that is, learning $3$ at resolution $1$ from $(1,2,6)$. 
For each resolution, the model encodes all grid cells and convert them to \textit{key vectors} with feed forward neural networks.
It then calculates scores for the next location given the prefix $(1,2,6)$ by performing a dot product with the query vector. 
This operation generates a probability distribution for each resolution (i.e., over $[4]$ and $[16]$), using the softmax function.
The model computes the cross-entropy loss based on this probability distribution and the actual value, i.e., $3$ at resolution $1$, 
and $10$ at resolution $2$.
Finally, the model calculates the gradients of the parameters $\theta$ from the summed cross-entropy losses. It then employs DP-SGD for model update. 

\subsection{Hierarchical Location Encoding}
HRNet introduces a novel approach in its location encoding mechanism, which has smaller number of parameters than the baseline model. 
The key feature is the employment of a hierarchical network structure with 2D quad deconvolutional network.

\subsubsection{The structure}
\label{sec:loc_arc}
The core of HRNet's location encoding is a trainable root vector, denoted as $\theta_{\rm root}$.
This root vector is augmented with a 2D quad deconvolutional network comprising of $d$ layers. 
The process involves the application of this 2D quad deconvolutional network to $\theta_{\rm root}$, producing $4^d$ distinct vectors. 
Essentially, this operation results in an embedding matrix of dimensions $(2^d,2^d,n_{\rm dim})$:
$
M_d={\rm deconv}^d(\theta_{\rm root}).
$
For a detailed explanation of the ${\rm deconv}$ operation, refer to Section~\ref{sec:deconvolution}.
Notably, the number of parameters in each deconvolutional layer only depends on the kernel size and is independent of $n_{\rm POI}$, sustaining a constant number of parameters. 
Consequently, this architecture requires only $O(d) = O(\log n_{\rm POI})$ parameters to generate $n_{\rm POI}$ vectors, which is a significant reduction compared to the $O(n_{\rm POI})$ parameters typically needed in embedding matrix.

\subsubsection{Vector assignment}
\label{sec:loc_assign}
The assignment of the $4^d$ vectors generated by the hierarchical location encoding component to the $n_{\rm POI}$ POIs follows a specific methodology.
Given that the grid size, $w*w$, is a power of 4, we set $d = \log_4(w*w)$. 
This approach facilitates an alignment of vectors with the grid, described as:
$$
{\rm hiencode}(l)=M_d[\lfloor l/w\rfloor, l\% w],
$$
where $l\in L$.
This assignment method naturally incorporates the first law of geography~\cite{tobler1970computer}, suggesting that geographically proximate locations have similar embeddings. 
This proximity principle is a direct result of the deconvolutional operation, where a single parent vector generates four adjacent embeddings.
\paragraph{\textcolor{black}{Scattered POI assignment}}
\textcolor{black}{
The above assignment assumed that POIs are the grid cells according to a uniform grid as described in Section~\ref{sec:staypoint} for simplicity. 
Note that this grid assumption is not necessary and HRNet can work with scattered POIs as long as the number of POIs is equal to or less than $4^d$. 
As mentioned above, we should follow the first law of geography rather than randomly assigning vectors to the scattered POIs.
Thus, we propose a method for assigning scattered POIs to vectors in $M_d$ following the proximity principle. 
First, we create a grid that covers all scattered POIs. 
Here, we set the minimum $d$ so that the number of POIs is less than $4^d$, and $w=2^d$. 
Then, we compute the coordinates of the center of each grid cell (${\rm cell}_j$) and assign each POI (${\rm POI}_{i}$) to a unique cell so that the Euclidean distance $d({\rm POI}_i, {\rm cell}_j)$ is as small as possible.
This assignment can be formalized as a linear assignment problem with a cost matrix $C$, where
$$
C=\{d({\rm POI}_i, {\rm cell}_j)\}_{i\in{[n_{\rm POI}]},j\in{[4^d]}}.
$$
This optimization can be solved using the Jonker-Volgenant algorithm~\cite{Crouse2016aerospace}. Once each POI is assigned to a unique grid cell, we can use the vector assignment method described above. Since POIs are aligned based on geographical proximity, this embedding also incorporates the proximity principle of the deconvolutional network.
}


\subsection{Scoring}
HRNet employs an efficient scoring mechanism that combines the hierarchical location encoding with dot product operations to establish a probability distribution over POIs.

\subsubsection*{Formulation of the Scoring Process}
In this process, the encoded prefix sequence $(v_1, \dots, v_{i-1})$, represented by $h_{i-1}$, is transformed into a query vector. 
This transformation is executed using a feed-forward neural network, denoted as $f_{\rm query}$. 
Concurrently, the encoded POI, ${\rm hiencode}(l)$, is processed into a key vector through another feed-forward neural network, named $f_{\rm key}$. 
Mathematically, this can be represented as
$
{\rm query}=f_{\rm query}(h_i)
$
and
$
{\rm key}_l=f_{\rm key}({\rm hiencode}(l)).
$
Following these transformations, the score of $l$ as $v_{i}$ is determined by computing the dot product between the query and the key vector
$
{\rm score}_l={\rm query}\cdot {\rm key}_l.
$
The final step involves applying the softmax function to these scores across all POIs, resulting in a probability distribution over POIs:
$
\Pr(v_{i}=l_{i}|v_1,\dots,v_{i-1})={\rm softmax}({\rm score}_{L})[l_{i}],
$
where ${\rm score}_{L}=({\rm score}_1,\dots,{\rm score}_{|L|})$.


\subsection{Multi-resolution Learning}
\label{sec:multi-task training}
In this section, we address the challenge of learning difficulty associated with a large number of labels by redefining the objective function in HRNet. 
The solution lies in employing a multi-resolution interpretation, facilitated by the hierarchical location encoding component. 
The essence of our approach is to extend beyond the original task by incorporating multiple tasks across various resolutions.
This multi-resolution strategy entails not only focusing on the primary task at hand but also simultaneously considering tasks at coarser resolutions, which are easier than the primary task. 
By doing so, the model gains a broader understanding of the data from the coarser resolutions, allowing it to better manage the complexity that comes with a large label space even in the noise of DP.

\subsubsection{Multi-resolution Interpretation}
Here, we maintain the assumption that the number of grid cells, denoted as $w\times w$, is a power of $4$, with $w=2^d$. 
This forms the basis for our multi-resolution interpretation, which extends beyond the primary grid division of $2^d\times 2^d$ (as described in Section~\ref{sec:problem formulation}).

We consider grid divisions of $2^{i_{\text{res}}}\times 2^{i_{\text{res}}}$ at each resolution level $i_{\text{res}}\in \{1,2,\dots, d-1\}$. 
Here, the $i_{\text{res}}$th resolution includes $L_{i_{\text{res}}}$ grid cells, represented as:
$
i_{\rm res}\in [d], L_{i_{\rm res}}:=[4^{i_{\rm res}}].
$
The matrix $M_{i_{\text{res}}}$ (refer to Section~\ref{sec:loc_arc}) encodes these grid cells by applying the vector assignment in Section~\ref{sec:loc_assign}. 
The encoding for grid cell $l$ at resolution $i_{\text{res}}$ is thus:
$$
{\rm hiencode}_{i_{\rm res}}(l)=M_{i_{\rm res}}[\lfloor l/ 2^{i_{\rm res}}\rfloor ,l \% 2^{i_{\rm res}}],
$$
where $l \in L_{i_{\text{res}}}$. 
This encoding is then fed into the scoring component to generate a probability distribution over the grid cells at resolution $i_{\text{res}}$:
$$
\Pr(l_{i_{\rm res}}=l|v_1,\dots,v_{i-1})={\rm softmax}({\rm score}_{L_{i_{\rm res}}})[l],
$$
where $l\in L_{i_{\rm res}}$ and ${\rm score}_{L_{i_{\rm res}}}=({\rm score}_1,\dots,{\rm score}_{|L_{i_{\rm res}}|})$.
Note that we use the query vector that is computed by GRU, which is common in all the resolutions.

\subsubsection{The Objective Function}
Our novel training strategy leverages the multi-resolution interpretation, incorporating multiple tasks across different resolutions. 
For a given grid cell $y$ in $L_d$ (the POI at the primary task level), the goal is to learn this grid cell and simultaneously learn its corresponding grid cells at lower resolutions $i_{\text{res}}$, where $i_{\text{res}} < d$.
That is, this approach introduces additional $d-1$ tasks into our training.

Each task at resolution $i_{\text{res}}$ comes with its own loss function:
$$
{\rm loss}_{i_{\rm res}}= \sum_{l\in L_{i_{\rm res}}} \delta_{l,{\rm up}_{i_{\rm res}}(y)} \log (\Pr(l_{i_{\rm res}}=l|v_1,\dots,v_{i-1})),
$$
where ${\rm up}_{i_{\rm res}}(y)\in L_{i_{\rm res}}$ is the grid cell covering $y$ at resolution $i_{\text{res}}$.
$\Pr(l_{i_{\rm res}}=l|v_1,\dots,v_{i-1})$ represents the inferred probability of the next grid cell being $l$ at resolution $i_{\rm res}$. 
These losses for all resolutions are then summed up, and this sum is used to compute the gradient for backward propagation.

Optimizing these loss functions facilitates the determination of the probability distribution for the next grid cell not only at the primary resolution $d$ but also at each additional resolution $i_{\rm res}$.
Moreover, the tasks defined at smaller $i_{\rm res}$ are inherently simpler due to the reduced number of labels $|L_{i_{\rm res}}|$.
Inferring finer cells from well-trained coarser resolution due to the simplicity, guided by the first law of geography, assists in mitigating the learning difficulty associated with the primary task.

\section{Private Pre-training}
\label{sec:pre-training}
Pre-training using publicly available data has been recognized as an effective strategy to mitigate the limitations of DP-SGD~\cite{amid2022public}.
However, the reliance on publicly accessible data is often a significant constraint, as such data may not always be available. 
To address this challenge, we introduce a novel private pre-training methodology that utilizes a DP compliant transition matrix, eliminating the need for public data.
This private pre-training strategy not only provides a `warm start' to accelerate the model's convergence but also enhances the model's capability to reduce empirical risk.

\subsection{DP Transition Matrix}
In our approach, we utilize a first-order transition matrix as ground truth for pre-training. 
The core task in this context is to predict the distribution of the next POI from a given POI. 
However, when dealing with a large number of POIs $n_{\rm POI}$, a direct POI-to-POI first-order transition matrix becomes excessively sparse. 
To avoid this issue, we leverage the multi-resolution interpretation of HRNet.
Instead of a direct POI-to-POI transition, we consider transitions from a broader region, specifically a grid cell at a coarser resolution $l_{i_{\rm res}}\in L_{{i}_{\rm res}}$, to a POI $l\in L_d$. 
\begin{equation}
\label{eq:tran}
{{\rm TRAN}_{i_{\rm res}}}[l_{i_{\rm res}}, l]=\sum_{\mathbf{v}\in D} \mathbf{1}_{V_{l_{i_{\rm res}},l}}(\mathbf{v})/|\mathbf{v}|,
\end{equation}
where $\mathbf{1}$ is an indicator function and $V_{l_{i_{\rm res}},l}$ is a subset of trajectory universe $\mathcal{V}$, which includes trajectories that have transition from $l_{i_{\rm res}}$ to $l$:
$
V_{l_{i_{\rm res}},l}:=\{\mathbf{v}\in\mathcal{V}|\exists i, {\rm up}_{i_{\rm res}}(\mathbf{v}[i])=l_{i_{\rm res}}\ \&\ \mathbf{v}[i+1]=l\}.
$
Note that this calculation normalizes the count with the sequence length $|\mathbf{v}|$, to limit sensitivity to $1$ for DP. 
To ensure DP, we introduce Laplace noise into the matrix as follows:
\begin{equation}
\label{eq:dptran}
{{\rm DPTRAN}_{i_{\rm res}}}[l_{i_{\rm res}}, l]={{\rm TRAN}_{\rm i_{\rm res}}}[l_{i_{\rm res}}, l]+{\rm Laplace}(\varepsilon),
\end{equation}
where ${\rm Laplace}(\varepsilon)$ is a Laplace noise that leads to $\varepsilon$-DP when sensitivity is $1$.
This approach results in a denser transition matrix, hence more informative and less noisy, which is useful in pre-training. Figure 1 left side shows a sample of such coarse transition matrix.

\subsection{Pre-training}
This section outlines the pre-training methodology using the DP transition matrix. 
During this phase, the core objective is to generate a probability distribution depicting the likelihood of transitions from a given region (a grid cell at resolution $i_{\rm res}$) to the POIs.
It is important to note that while both the hierarchical location encoding component and its scoring component undergo pre-training, the prefix encoding component (GRU) (see Figure 2) does not since the DP transition matrix does not contain prefix information.

\subsubsection{The Architecture}
The architecture for the pre-training phase incorporates a temporary substitute for GRU, necessitated by the absence of prefix information in the DP transition matrix. 
This temporary component is designed to mimic the function of the GRU, generating a query vector with dimensions identical to those of an encoded prefix.
Conceptually, this query vector symbolizes the initial grid of the transition, as represented by the following equation:
\begin{equation}
\label{eq:temporal_component}
h={\rm temp}(x),
\end{equation}
where $x$ is the input vector that represents the given region, with further details provided in the following section. 
Notably, this temporary component is exclusive to the pre-training phase and is replaced with GRU during model training.

Aside from this modification, the architecture during pre-training remains consistent with the HRNet model.
The scoring component processes the vector $h$ output by the temporary component, in tandem with POI encodings obtained from the hierarchical encoding component. 
The query vector is computed as ${\rm query}=f_{\rm query}(h)$, and the key vector for a potential POI $l$ is computed as ${\rm key}_l=f_{\rm key}({\rm hiencode}_d(l))$. The score for POI $l$ is then calculated using a dot product ${\rm score}_l={\rm query}\cdot {\rm key}_l$.
The application of a softmax function to these computed scores results in a probability distribution over $L$.

\subsubsection{The Objective Function}


To enhance learning, we propose  a data augmentation approach to augment the transition matrix, achieved through  proportional blending of regions (grid cells at resolution $i_{\rm res}$).
A mixing ratio vector is defined as:
$
\mathbf{r}=(r_1,\dots,r_{|L_{i_{\rm res}}|})\in [0,1]^{|L_{i_{\rm res}}|},
$
subject to the constraint $\sum_{i=1}^{|L_{i_{\rm res}}|} r_i=1$. 
Here, each $r_i$ denotes the proportion of region $i$ in the initial state, forming a prior distribution as a multinomial distribution with probabilities $\mathbf{r}$.
Then, the ground truth probability distribution given $\mathbf{r}$ is calculated as:
$$
\Pr (l_{\rm next}|\mathbf{r})=\sum_{l\in L_{i_{\rm res}}}{\rm Mult}(l_{\rm initial}=l|\mathbf{r}) \Pr(l_{\rm next}|l_{\rm initial}=l),
$$
where ${\rm Mult}(\cdot|\mathbf{r})$ is the multinomial distribution with probabilities $(r_1,\dots,r_{|L_{i_{\rm res}}|})$ and $\Pr(l_{\rm next}|l_{\rm initial}=l)={\rm DPTRAN}_{i_{\rm res}}[l]$.
This essentially models the probability distribution for the next POI, given a stay at $l_{\rm initial}$ with the probability $\Pr (l_{\rm initial}|\mathbf{r})$.
The mixing ratio is generated by sampling from a Dirichlet distribution
$
\mathbf{r}\sim {\rm Dirichlet}(\mathbf{a}),
$
where $\mathbf{a}=(1,1,\dots,1)$ and $|\mathbf{a}|=|L_{i_{\rm res}}|$.  Note that the augmented transition matrix is still DP due to the post processing property. 

For the model derived probability distribution, we incorporate the mixing ratio $\mathbf{r}$ into the location encoding as follows:
$$
x=\sum_{l\in L_{i_{\rm res}}} {\rm Mult} (l_{\rm initial}=l|\mathbf{r}){\rm hiencode}_{i_{\rm res}}(l).
$$
This is the weighted sum of the encoding of region $l\in L_{i_{\rm res}}$ with ratio $\mathbf{r}$.
Then, we use the mixed encoding $x$ as input to the temporal component (i.e., Equation~(\ref{eq:temporal_component})) to derive the key vector and then the output distribution.
Finally, the loss is
\begin{equation}
\label{eq:loss}
{\rm loss}={\rm KL}(\Pr (l_{\rm next}|\mathbf{r})||\hat{y}|\mathbf{r}),
\end{equation}
where $\hat{y}|\mathbf{r}$ is the output distribution derived from the key vector.
With this objective function, our model not only learns the probability distributions of ${\rm DPTRAN}$ but also learns from a continuous representation among these probability distributions. 

\subsubsection{Discussion}
\label{sec:discussion}
In summary, HRNet is pretrained using {\rm DPTRAN} (Equation~\ref{eq:dptran}) based on the KL loss (Equation~\ref{eq:loss}). Here, we further explain the rationale behind this choice from two perspectives:
\paragraph{Warm Start}
One challenge of DP-SGD is that it requires many iterations due to starting from a cold state. A warm start helps to mitigate this issue~\cite{amid2022public}, and we anticipate this benefit in our approach. 
This is because when the KL loss during pre-training is optimized, it also helps to optimize the cross-entropy loss during main training due to the multi-task learning.
It is important to note that optimizing cross-entropy loss is essentially the same as optimizing KL loss. Although marginalization to generate DPTRAN means the KL loss is not exactly the original cross-entropy loss, we expect it to be similar enough to provide a warm start.
\paragraph{Differential Privacy Efficiency}
There are several alternatives for pretraining beyond the first-order transition matrix {\rm DPTRAN}, such as second-order transitions and transitions that include time information. Additionally, we could use a sophisticated and nonuniform grid based on density, similar to privtree~\cite{zhang2016privtree}, for the initial grid of {\rm DPTRAN}. As shown by Zhou et al.~\cite{Zhou2021BypassingTA}, using public data to find the gradient subspace addresses the issue of a large number of parameters and improves noise stability. Therefore, a more precise {\rm DPTRAN} as mentioned above could potentially enhance the main training results. However, these methods typically require additional privacy budget.
From our experience, we found that using first-order transition and a fixed resolution ($i_{\rm res}=2$) is more beneficial than these complex methods. 
First-order transitions are denser than second-order transitions and those with time information, making the perturbed information more useful for pre-training. 
A fixed resolution allows us to allocate a larger portion of the privacy budget to perturb the transition matrix, which provides better information for pre-training.
Moreover, a fixed resolution enables algorithmic privacy budget allocation (see Section~\ref{sec:priv_budget_allocation}) without seeing raw data.

Although the selection of first-order transition and ($i_{\rm res}=2$) was heuristically made based on the principles mentioned above, it is not theoretically clear why this approach is effective. Understanding what type of information or loss contributes to the main task is challenging, and we believe this study provides some insights. Investigating the types of information that are beneficial for pre-training is an interesting direction for future research.

\subsection{Privacy Analysis of HRNet with Pre-training}
\label{sec:privacy-analysis}
We analyze the privacy guarantees of HRNet when pre-trained using the DP transition matrix.
\textcolor{black}{
\begin{theorem}
\label{theo:privacy_analysis}
    Given a dataset $D$, privacy parameters $\varepsilon_1, \varepsilon_2\in\mathbb{R}^+$ and $\delta\in[0,1]$, and initial parameters of HRNet $\theta$, the pre-training and training process for HRNet is as follows:
    \begin{equation}
    \label{eq:pre-training-phase}
    \theta^\prime = {\rm pretrain}(\theta, {\rm dptran}(D,\varepsilon_2)).
    \end{equation}
    \begin{equation}
    \label{eq:training-phase}
    \theta^{\prime\prime} = {\rm DPSGD}(D, \theta^\prime, \varepsilon_1, \delta).
    \end{equation}
    Here, ${\rm dptran}(D,\varepsilon_2)$ refers to the computation of Equation~(\ref{eq:dptran}) and {\rm pretrain} is the standard  training with loss function~(\ref{eq:loss}).
    {\rm DPSGD} is the private training algorithm described in Section~\ref{sec:differential-privacy}, utilizing the allocated privacy parameters $(\varepsilon_1,\delta)$, and initial parameter $\theta^\prime$.
    Deriving the trained parameters for HRNet (i.e., $\theta^{\prime\prime}$) satisfies $(\varepsilon_1+\varepsilon_2,\delta)$-DP.
\end{theorem}
\begin{proof}
    We can see that the entire computation is the sequential composition of Equation~(\ref{eq:pre-training-phase}) and Equation~(\ref{eq:training-phase}).
    For Equation~(\ref{eq:pre-training-phase}) (i.e., the pre-training phase), the raw data $D$ is employed solely for generating DPTRAN (i.e., Equation~(\ref{eq:dptran})). 
    The application of the Laplace mechanism, which is known to satisfy $\varepsilon_2$-DP~\cite{dwork2006calibrating}, ensures that the pre-training phase is $\varepsilon_2$-DP. 
    According to the post-processing theorem~\cite{dwork2014algorithmic}, computing $\theta^\prime$ satisfies $\varepsilon_2$-DP since it does not further utilize $D$.
    For Equation~(\ref{eq:training-phase}) (i.e., the training phase), the raw data $D$ is used for {\rm DPSGD} and DPSGD satisfies $(\varepsilon_1,\delta)$-DP (see Section~\ref{sec:differential-privacy} for detail).
    By applying the composition theorem of DP~\cite{dwork2014algorithmic}, deriving $\theta^{\prime\prime}$ satisfies $(\varepsilon_1+\varepsilon_2,\delta)$-DP.
\end{proof}
}

\subsubsection{Privacy Budget Allocation}
\label{sec:priv_budget_allocation}
Based on the analysis above, we have demonstrated that HRNet with pre-training satisfies $(\varepsilon_1 +\varepsilon_2,\delta)$-DP. 
A critical challenge arises in optimally allocating the given privacy budget $\varepsilon$ between $\varepsilon_1$ and $\varepsilon_2$ due to the difficulty in hyperparameter selection, such as cross-validation, under the DP constraint.
An excessive allocation to $\varepsilon_2$ (thus reducing $\varepsilon_1$) may lead to premature depletion of the privacy budget for DP-SGD during model training, potentially hindering the model's convergence and impairing its final performance. Conversely, a small $\varepsilon_2$ risks yielding an insufficient pre-trained model, thereby diminishing its effectiveness in supporting DP-SGD.
Striking a balance in privacy budget distribution is therefore crucial for enhancing the overall performance of the model. 
This section proposes a heuristic solution to this budget allocation challenge.

\paragraph{Solution}
Assume that we use the $i_{\rm res}$th resolution for the initial region of the DP transition matrix.
We propose the following allocation formula:
\begin{equation}
\label{eq:privacy budget allocation solution}
\varepsilon_2=\min \left(\frac{cw^24^{i_{\rm res}}\log (w)}{|D|},\varepsilon\right), \varepsilon_1=\varepsilon-\varepsilon_2.
\end{equation}
This solution is derived (see Appendix~\ref{sec:privacy budget allocation derivation}) in order to maintain a constant signal-to-noise ratio (SNR) in the DP transition matrix, based on the dataset's meta-information (the number of records $|D|$ and the parameter that decides the number of grid cells $w$) which is not sensitive, thus publicly available.
Our aim is to consistently ensure a minimum quality of the DP transition matrix for effective pre-training.

\section{Experiments}
Our experiments aim to evaluate the efficacy of HRNet in generating human mobility data under DP constraints. 
We first compare the privacy-utility trade-off of our approach with existing state-of-the-art methods. 
Additionally, we conduct an ablation study to assess the impact of our novel components: hierarchical location encoding, multi-task learning, and private pre-training.

\subsection{Setup}
\label{sec:exp_setup}
\subsubsection{datasets:}
We utilize five datasets, broadly classified into three categories: human mobility, taxi trajectory (road network), and synthetic data. 
Human mobility encompasses a wide range of unrestricted mobility patterns, while road network data is constrained by road networks.
The datasets employed are Geolife~\cite{zheng2009mining} (human mobility), Peopleflow~\cite{sekimoto2011pflow} (human mobility), Didi in Chengdu~\cite{wang2022deep} (taxi trajectory), and two synthetic datasets. 
The synthetic datasets are used in the ablation study and discussed in detail in Section~\ref{sec:exp_ablation}. 
We use randomly sampled $10,000$ trajectories from each dataset as the training dataset.
See details of the datasets in Appendix~\ref{sec:dataset}.




\subsubsection{Preprocessing:}
\label{sec:preprocess}
In the preprocessing phase, our approach centers on the concept of \textit{stay points}, as defined by Li et al.~\cite{li2008mining}. 
A stay point is identified when an individual remains within a radius of $m$ meters for a duration exceeding $t$ minutes. 
To prepare our dataset, we first process each trajectory to identify these stay points. 
The identified stay points are then regarded as critical waypoints, representing significant stops or areas of interest in the movement patterns of individuals. 
By focusing on these waypoints, we can distill the essence of each trajectory, thereby capturing the most relevant and informative aspects of human mobility.
Note that our preprocessing inevitably leads to a loss of detailed information, such as minor route variations and brief stops.

\subsubsection{Competitors}
\label{sec:competitors}
We compare our method with a baseline and three state-of-the-art methods in  distinct categories as outlined in Introduction:

\begin{itemize}
\item 
\textbf{Baseline}:
implements the generator as described in Section~\ref{sec:baseline} with DP-SGD. 
    \item 
\textbf{PrivTrace ~\cite{wang2023privtrace}}: a state-of-the-art Markov-based approach leveraging Markov models for mobility data analysis.
\item \textbf{Clustering ~\cite{liu2021differentially}}: a generative method with mobility data clustering using $k$-means.
\item 
\textbf{MTNet ~\cite{wang2022deep}}: a deep learning method specifically tailored for road network data.
\end{itemize}

We excluded tree-based methods such as DPT~\cite{he2015dpt} from our comparison. 
This decision was based on reports~\cite{wang2023privtrace, liu2021differentially} from the aforementioned studies, which indicated superior performance of the Markov and the clustering-based methods over tree-based approaches. 
For our implementation, we relied on publicly available open-source codes for PrivTrace and MTNet. 

\begin{table*}[]
\caption{The discrepancy of the generated dataset to the real dataset (Geolife / Peopleflow) with $w=32$ when $\varepsilon=2$ and $\delta=10^{-5}$.}
\label{tab:main_result}
\begin{tabular}{c|cccccc|cc}
\multicolumn{1}{l|}{} & \multicolumn{6}{c|}{JS Divergence $\downarrow$}                                                                                 & \multicolumn{2}{c}{ARE $\downarrow$}       \\
                      & waypoint    & destination   & transition    & travel distance  & {\textcolor{black}{diameter}} & \textcolor{black}{density\_t}      & \textcolor{black}{traj density}  & \textcolor{black}{traj pattern}  \\ \hline
Clutering             & 9.92 / 12.9 & 0.438 / 0.650 & 0.422 / 0.542 & 0.0227 / 0.154   & 0.0921 / 0.133                  & 0.123 / 0.0813               & 0.233 / 0.852 & 0.701 / 0.772 \\
PrivTrace             & 12.6 / 11.0 & 0.274 / 0.649 & 0.245 / 0.431 & 0.0476 / 0.0181  & 0.0712 / 0.0531                 & -               & 0.212 / 0.835 & 0.712 / 0.781 \\
Baseline              & 6.68 / 10.1 & 0.397 / 0.670 & 0.397 / 0.486 & 0.0275 / 0.151   & 0.0932 / 0.303                  & 0.115 / 0.0512  & 0.382 / 0.811 & 0.771 / 0.912 \\
HRNet         & 5.67 / 6.17 & 0.192 / 0.586 & 0.212 / 0.309 & 0.00871 / 0.0569 & 0.0681 / 0.183                  & 0.0523 / 0.0425 & 0.184 / 0.562 & 0.671 / 0.743
\end{tabular}
\end{table*}

\subsubsection{Evaluation metrics:}
To compare the quality of generated mobility data, we evaluate the discrepancies between generated and original stay-point trajectories (denoted as $\mathbf{v}$) across various metrics. 
We employ the Jensen-Shannon (JS) divergence and average relative error (ARE) for measuring discrepancy (both the lower the better), with details provided in the Appendix~\ref{sec:evaluation_metrics}. 
The key metrics considered are as follows which are all probability distributions.
\textbf{Waypoint} captures which waypoints are traversed given a starting point.
\textbf{Destination} captures which location serves as the final destination given a starting point.
\textbf{Transition} captures where the trajectory moves next from a starting point.
\textbf{Travel distance} represents the total length of a trajectory.
\textbf{Diameter} measures the maximum distance between any two points along a trajectory.
\textbf{Route} captures the route taken from a starting point.
\textbf{Density at $t$} is the probability of a trajectory being at a specific location at time $t$.
\textbf{Trajectory density} is the probability of a trajectory passing through a specific area.
\textbf{Trajectory pattern} captures the most frequent  transition patterns.

\subsection{Comparison with Baseline, PrivTrace, and Clustering}
In this section, we present a comparative evaluation of HRNet against existing approaches for human mobility not constrained by road networks. 
The competitors in this analysis include the Baseline, PrivTrace, and Clustering methods.
The used datasets are the Geolife dataset and the Peopleflow dataset.

\subsubsection{Main result}
Table~\ref{tab:main_result} showcases the main result, where we assess the discrepancy across four metrics with a fixed privacy budget ($\varepsilon=2$). 
Notably, our proposed method demonstrates significant superiority across all metrics compared to the other approaches. 
The statistical-based methods, such as Clustering and PrivTrace, inherently struggle with a high signal-to-noise ratio due to the need to add noise to sparse data. 
Conversely, deep learning methods have the potential to uncover hidden characteristics of POIs, moving beyond mere sparse information.
The Baseline method experiences a decrease in utility, largely due to its incompatibility with DP-SGD, as it still relies on learning sparse information through an embedding matrix. 
Our approach, incorporating a deconvolutional network, multi-task training, and private pre-training, effectively learns dense, hidden characteristics in a manner that harmonizes with DP-SGD.

\subsubsection{Utility-privacy trade-off}

This subsection delves into the utility-privacy trade-off of HRNet in comparison to other competitors, achieved by varying the privacy budget. 
The results depicted in Figures~\ref{fig:comp_destination} and \ref{fig:comp_waypoint} illustrate this analysis. 
It is observed that at lower privacy budgets (e.g., $\varepsilon<1$), PrivTrace and clustering sometimes exhibit better performance. 
This is attributed to their direct computation of transitions.

In contrast, deep learning methods generally require extensive training epochs to learn the hidden features of human mobility. However, our method enhances the learning process through pre-training and multi-task learning, effectively reducing the need for a larger privacy budget, unlike the baseline method. 
Furthermore, as epsilon increases, deep learning approaches increasingly outperform their counterparts. 
This advantage stems from their direct learning of raw data, avoiding the utility loss associated with sensitivity bounding, a common issue in statistical-based methods like clustering and PrivTrace.

\subsection{Comparison with MTNet}
MTNet~\cite{wang2022deep} represents a leading deep learning methodology designed specifically for generating short trips within road networks. 
The architecture of MTNet necessitates that trajectories be confined to connected road networks. 
This requirement is incompatible with datasets like Geolife and Peopleflow, which encompass mobility patterns extending beyond road networks, including activities like train travel or park traversing.
Hence, for a comparative analysis with MTNet, we utilize the Didi dataset.

\begin{figure*}[t]
    \centering
    \begin{minipage}{.47\columnwidth}
        \centering
        \includegraphics[width=\linewidth]{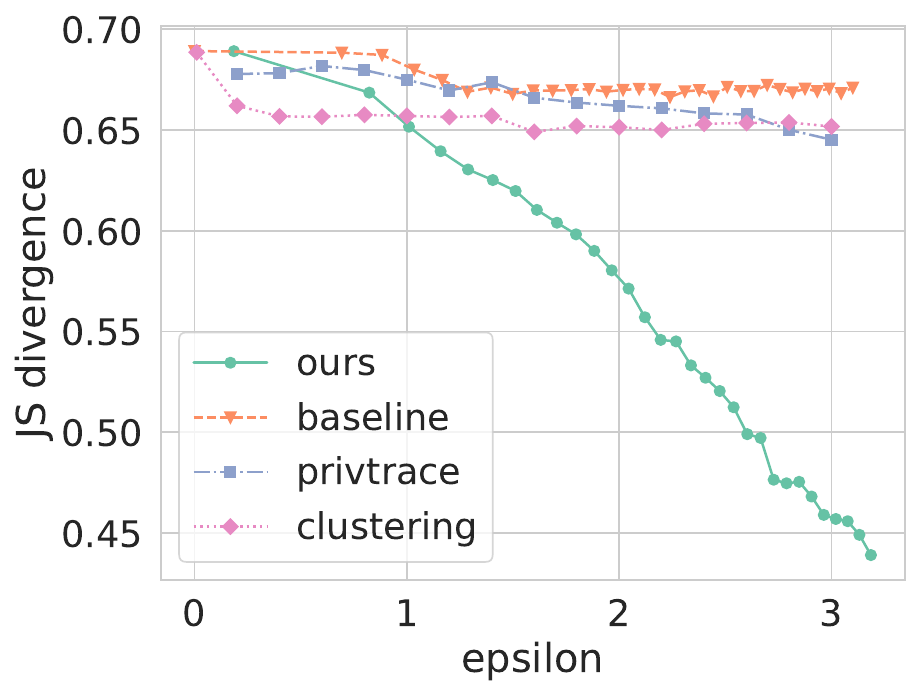}
        \caption{The discrepancy of destination on Peopleflow dataset ($w=64$) for each $\varepsilon\in [0,3.0]$.}
        \label{fig:comp_destination}
    \end{minipage}%
    \hspace{10pt}
    \begin{minipage}{.47\columnwidth}
        \centering
        \includegraphics[width=\linewidth]{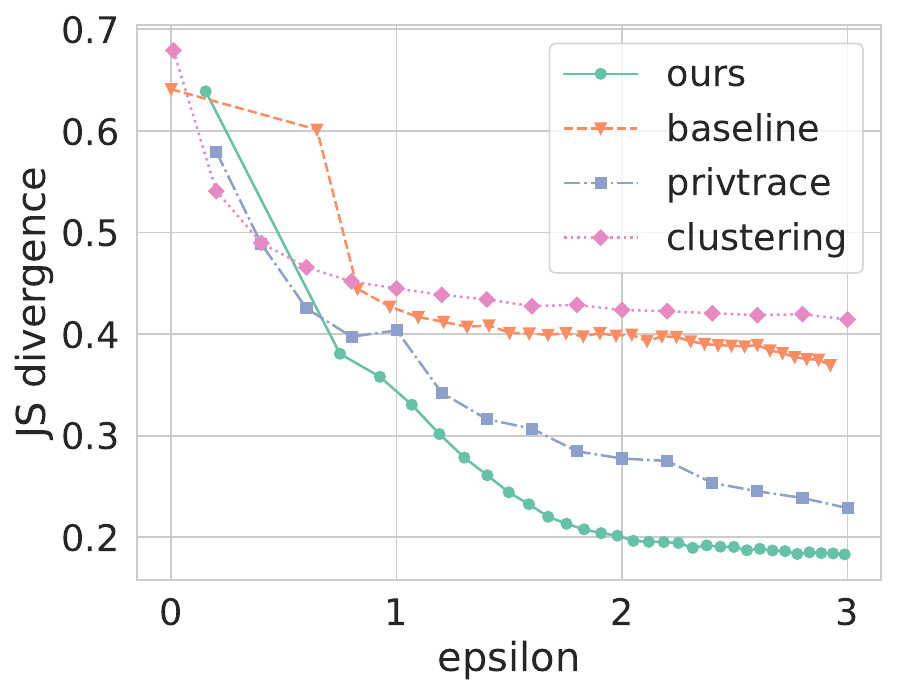}
        \caption{The discrepancy of transition on Geolife dataset ($w=32$) for each $\varepsilon\in [0,3.0]$.}
        \label{fig:comp_waypoint}
    \end{minipage}
    \hspace{10pt}
    \begin{minipage}{.47\columnwidth}
        \centering
        \includegraphics[width=\linewidth]{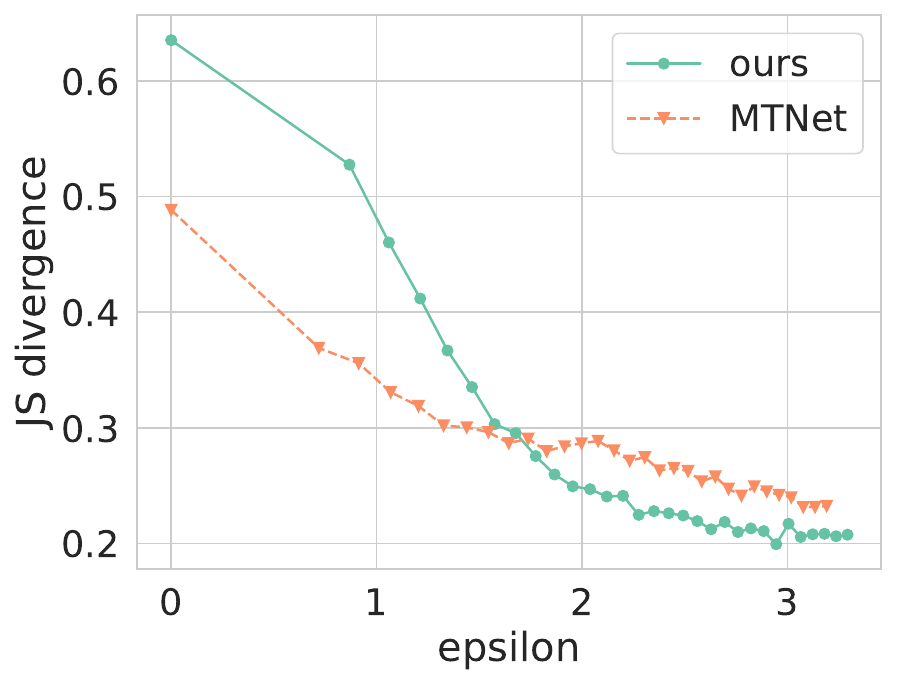}
        \caption{The discrepancy of destination on Didi dataset ($w=32$) for each $\varepsilon$ and $\delta=10^{-5}$.}
        \label{fig:mtnet_destination}
    \end{minipage}%
    \hspace{10pt}
    \begin{minipage}{.47\columnwidth}
        \centering
        \includegraphics[width=\linewidth]{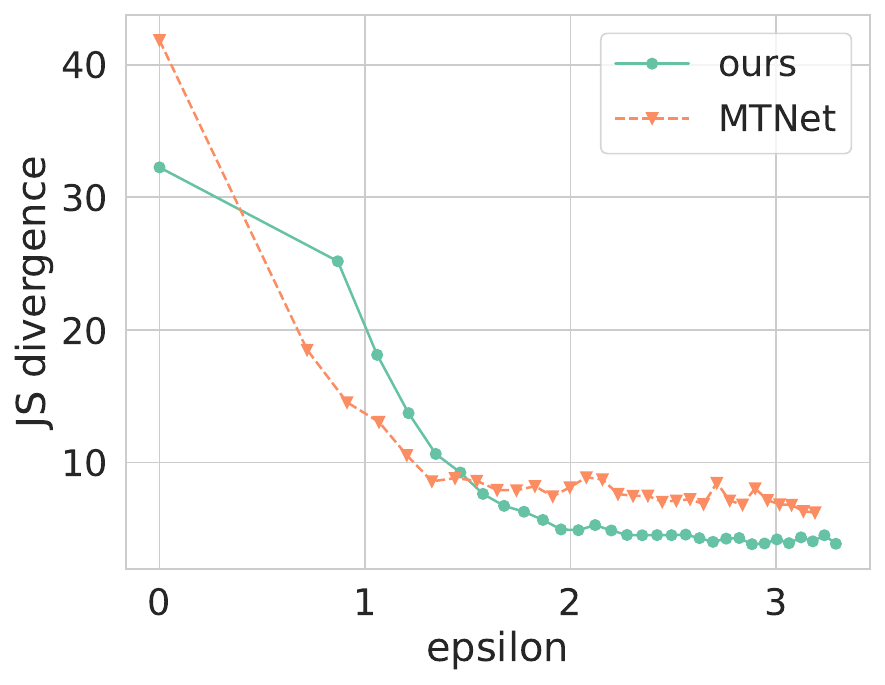}
        \caption{The discrepancy of route on  Didi dataset ($w=32$) for each $\varepsilon$ and $\delta=10^{-5}$.}
        \label{fig:mtnet_route}
    \end{minipage}
\end{figure*}

Our results, as illustrated in Figures~\ref{fig:mtnet_destination} and \ref{fig:mtnet_route}, focus on evaluating destination and route discrepancies. 
For both metrics, MTNet initially exhibits superior performance, primarily because it generates trajectories using road network metadata. 
However, as epsilon increases, our method demonstrates improved results. 
This discrepancy arises from the differing objectives of the training targets. 
Our method is tailored to learn semantically significant next waypoints. 
Conversely, MTNet is designed to directly learn the next road segment.
As a result, MTNet is required to learn a longer sequence which is difficult to learn under DP constraints.
Hence, our method has better convergence in the case where DP-SGD is applied.


\subsection{Ablation study}
\label{sec:exp_ablation}
In this section, we conduct an ablation study to elucidate the individual contributions of the key components of HRNet: the hierarchical location encoding via deconvolutional network, multi-resolution multi-task learning, and private pre-training. 
Our analysis begins with the Geolife dataset, followed by an examination using synthetic datasets, each with distinct characteristics.

\textbf{Random Dataset:}
The Random dataset comprises of trajectories, each with a fixed length $2$. 
The initial and second locations in these trajectories are randomly selected from all POIs, forming a total of $10,000$ trajectories. 
This dataset is characterized by completely random empirical transition distributions, effectively disregarding the first law of geography.

\textbf{Straight Dataset:}
Each trajectory in the Straight dataset is of fixed length $3$. 
The first location is chosen from even columns on a grid (i.e., $L_1=\{l\in [w*w]|\lfloor l/w\rfloor /2=0\}$). 
Subsequently, the second and third locations are determined as the next two rows above in the same column ($l_2=l_1+w$ and $l_3=l_2+w$). 
This study creates $10,000$ trajectories.

The ablation study is structured around six configurations, combining the elements of deconvolutional network, multi-task training, and private pre-training:
\begin{itemize}
    \item 
\textbf{Baseline:} baseline generator with DP-SGD.
\item \textbf{Pre-training:} baseline + private pre-training (Section~\ref{sec:pre-training}). Unlike our complete model, the location encoding component is not pre-trained here due to the absence of coarse region embeddings.
\item \textbf{Deconvolutional Network:} the location encoding component of the baseline is replaced with a 2D quad deconvolutional network.
\item \textbf{Deconvolutional Network + Pre-training:} this combines private pre-training with the deconvolutional network setup.
\item \textbf{Deconvolutional Network + Multi-task Training:} this integrates multi-task training (Section~\ref{sec:multi-task training}) with the deconvolutional network framework.
\item \textbf{Deconvolutional Network + Multi-task Training + Pre-training (HRNet):} the complete solution. 
\end{itemize}

Through this study, we aim to isolate and understand the contribution of each component to the overall efficacy of our method in capturing and replicating human mobility patterns.

\subsubsection{Real data}
\label{sec:ablation_real}
\begin{figure}[t]
    \centering
    \begin{minipage}{.47\columnwidth}
        \centering
        \includegraphics[width=\linewidth]{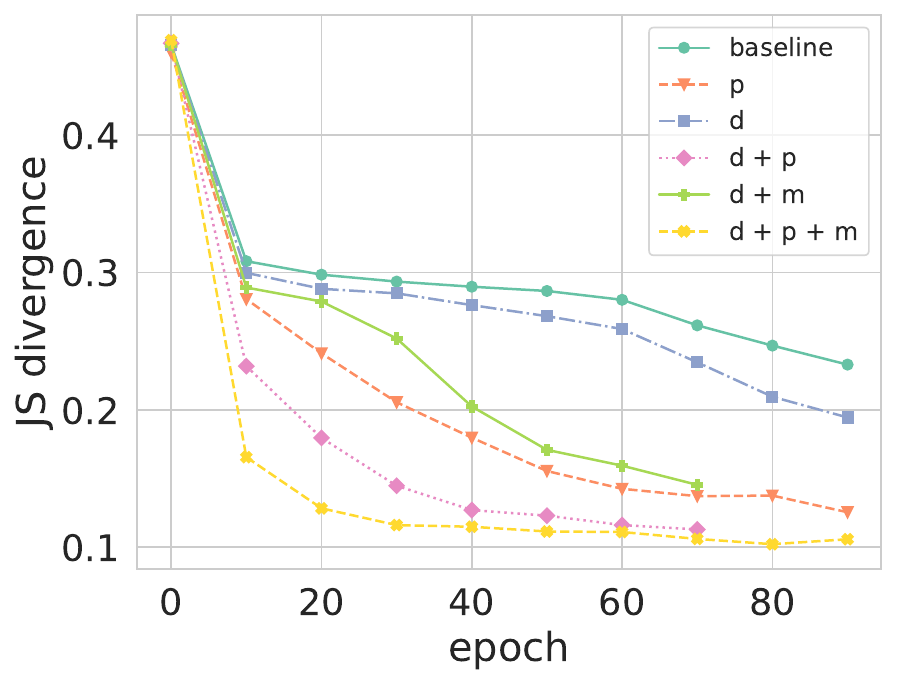}
        \caption{The discrepancy of  transition on Geolife  ($w=64$). 
        $p$, $d$, and $m$ represent pre-raining, deconvolutional network, and multi-task learning, respectively.}
        \label{fig:ablation_real}
    \end{minipage}%
    \hspace{10pt}
    \begin{minipage}{.47\columnwidth}
        \centering
        \includegraphics[width=\linewidth]{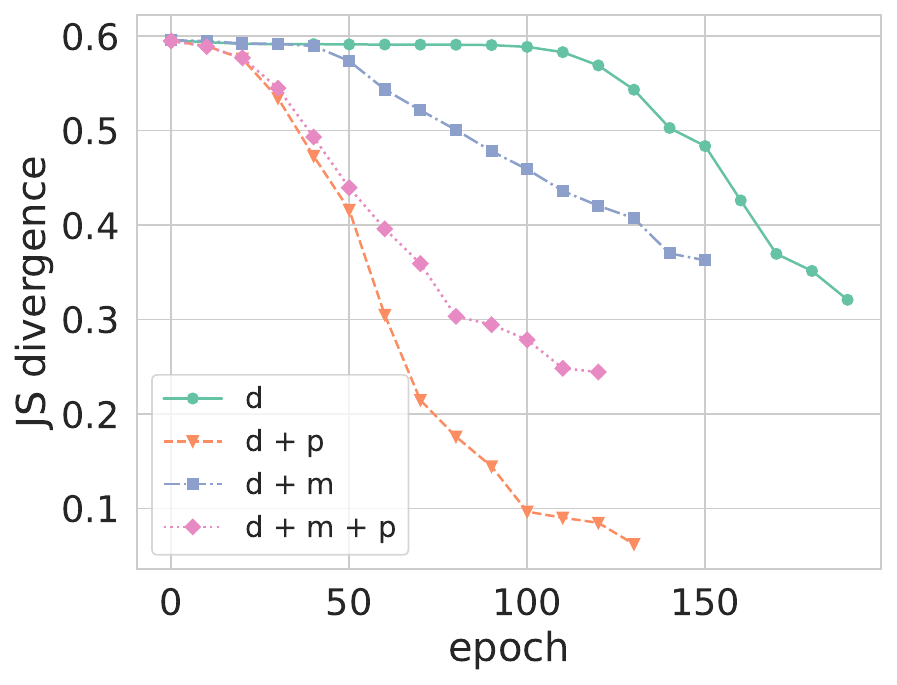}
        \caption{The discrepancy of transition on  Random dataset ($w=32$). $d$, $m$, and $p$ mean deconvolutional network, multi-task learning, and pre-training, respectively.}
        \label{fig:ablation_multi-task}
    \end{minipage}
    \hspace{10pt}
\end{figure}

The evaluation on the Geolife dataset with a grid size of $w=64$ is depicted in Figure~\ref{fig:ablation_real}. 
The results indicate that the incorporation of each component – the deconvolutional network, multi-task learning, and pre-training – not only enhances the accuracy in terms of discrepancy reduction but also accelerates convergence.

The deconvolutional network, in particular, contributes to improvement of discrepancy by offering an efficient parameterization. 
Its architecture supports the integration of multi-task learning and pre-training while providing benefits even when deployed independently. 
Multi-task learning improves the convergence speed with DP-SGD by incorporating simpler tasks as observed in the study by Dockhorn et al.\cite{dockhorn2022differentially}. 
Pre-training, on the other hand, offers a substantial boost in convergence speed through a warm start and enhances utility, aligning with findings by Amid et al.~\cite{amid2022public}.

\subsubsection{Deconvolutional Network}
\label{sec:exp_deconv}
\begin{table}[]
\caption{\textcolor{black}{The number of parameters with varying grid size $w$.}}
\label{tab:n_parameters}
\begin{tabular}{@{}ccccc@{}}
\toprule
         & $w=8$    & $w=16$   & $w=32$   & $w=64$    \\ \midrule
baseline & $14,054$  & $26,534$ & $76,454$ & $276,134$ \\
ours     & $24,710$ & $28,934$ & $41,606$ & $47,942$  \\ \bottomrule
\end{tabular}
\end{table}
\textcolor{black}{
The deconvolutional network not only enables multi-task learning and private pre-training but also reduces the number of parameters. As mentioned in Section~\ref{sec:bottleneck}, it is known that the lower bound of empirical risk increases as the number of parameters increases~\cite{bassily2020stability}. Therefore, the reduced number of parameters is a crucial advantage of the deconvolutional network, and we explore this aspect here.
Table~\ref{tab:n_parameters} compares the number of parameters with the baseline architecture (i.e., matrix embedding). The number of parameters in HRNet depends logarithmically on the number of POIs, in contrast to the baseline mechanism, which depends linearly. 
As a result, even if $w$ becomes larger, the number of parameters remains stable. 
This difference theoretically improves stability, and this is supported by the experimental results shown in Figure~\ref{fig:ablation_real}.
}

\subsubsection{Multi-task learning}
\label{sec:exp_multi_task}

HRNet benefits from the efficient utilization of inferences derived from simpler tasks, which align with the first law of geography.
\textcolor{black}{
It is important to note that bias based on the first law of geography can also have a negative impact.}
To further substantiate this point, we investigate a scenario where the first law of geography is not applicable, utilizing the Random dataset for this purpose. 
The results of this investigation are presented in Figure~\ref{fig:ablation_multi-task}.
Contrary to previous observations, the addition of multi-task learning in this context results in a decrease in utility. 
This decline can be attributed to the inductive bias towards the first law of geography, which, in the case of the Random dataset, leads to a misalignment with the actual data characteristics. 
Consequently, the multi-task learning framework, while beneficial under certain conditions, may  impact the results negatively when the underlying data does not adhere to the geographic principles it assumes. 

\textcolor{black}{
Even for datasets that generally follow the first law of geography, it is possible for some grid cells to not satisfy this law. 
For example, a grid cell near the border of a coarser cell tends to differ from the characteristics of the coarser cell because people in this grid cell often move to the adjacent coarser cell, whereas people in grid cells near the center of the coarser cell tend to stay within the same coarser cell. 
It is important to note that some cells may suffer due to this bias, even though it globally improves performance. 
Solving this aspect further is an interesting direction for future research.
}
\begin{table}[]
\caption{\textcolor{black}{Results for next location prediction and trajectory classification. We used the same privacy budget for the baseline model and HRNet $(\varepsilon=3.0,\delta=10^{-5})$.}}
\label{tab:classification_table}
\begin{tabular}{@{}lccccc@{}}
\toprule
                                                                                                           &                          & Acc@1 & Acc@5 & F-score & AUROC \\ \midrule
\multicolumn{1}{c|}{\multirow{3}{*}{\begin{tabular}[c]{@{}c@{}}next \\loc\\ pred\end{tabular}}}   & baseline w/o DP          & 29.7  & 45.1  & 40.9      & 84.5  \\
\multicolumn{1}{c|}{}                                                                                      & baseline                 & 13.3  &30.6   & 25.2      & 60.3  \\
\multicolumn{1}{c|}{}                                                                                      & HRNet                     & 23.1 & 37.7  & 30.2      & 70.2  \\ \midrule
\multicolumn{1}{l|}{\multirow{4}{*}{\begin{tabular}[c]{@{}l@{}}traj \\ class\\ \end{tabular}}} & Random Forest & 78.2  & -     & 74.3      & 92.6  \\
\multicolumn{1}{l|}{}                                                                                      & baseline w/o DP          & 73.4  & -     & 70.1      & 91.3  \\
\multicolumn{1}{l|}{}                                                                                      & baseline                 & 48.3  & -     & 47.8      & 72.1  \\
\multicolumn{1}{l|}{}                                                                                      & HRNet                     & 58.3  & -     & 55.9      & 78.7  \\ \bottomrule

\end{tabular}
\end{table}
\textcolor{black}{
\subsection{Application Case Study}
\label{sec:application_case_study}
}
\subsubsection{\textcolor{black}{Next location prediction}}
\label{sec:exp_next_location_prediction}
\textcolor{black}{
Next location prediction or recommendation is crucial for various applications such as smart cities, tourism, and advertising. 
In this section, we evaluate our model’s ability to predict the next location. 
To ensure a more realistic setting, we preprocess the Geolife dataset with the notion of significant locations~\cite{Martin_2023_trackintel, hong2023context, yang2014predicting} to make scatterd POIs, instead of a grid-based approach. 
Details on embedding scattered POIs instead of grid cells can be found in Section~\ref{sec:loc_assign}.
This preprocessing results in $734$ POIs and $9,816$ trajectories, which are split into a training set and a test set with a ratio of $0.9$ to $0.1$. 
We trained generative models on the training dataset and performed next location inference on the test dataset.
Table~\ref{tab:classification_table} presents the results. Acc@$k$ represents the accuracy when the correct answer is included in the top-$k$ predictions, and AUROC is the area under the receiver operating characteristic curve.
The baseline w/o DP is trained using SGD instead of DP-SGD. 
We spent the same privacy budget $(\varepsilon=3,\delta=10^{-5})$ to both the baseline and HRNet. 
The results show that while the baseline model significantly loses utility due to noise, HRNet outperforms the baseline, demonstrating its robustness to noise in the case of scattered POIs.
}
\subsubsection{\textcolor{black}{Trajectory classification}}
\label{sec:trajectory_classification}
\textcolor{black}{
Trajectory classification involves training a model to categorize trajectories into different types.
Here, we consider classification of transportation modes, a well-known problem in trajectory classification~\cite{zheng2010understanding, xiao2017identifying, liu2019spatio}. 
Some of trajectories in the Geolife dataset are labeled as bike, walk, car, bus, and train. 
We create $5,980$ trajectories using grid-based discretization with $w=32$, splitting them into a training set and a test set with a $0.9$ to $0.1$ ratio.
Using the training set, we first train a generative model (HRNet or the baseline) to synthesize trajectories. 
Here, the label information is incorporated into the generative model by extending Equation~\ref{eq:baseline_encode} as follows:
$
{\rm encode}(v_i)=[M_{\rm time}[t_i], M_{\rm POI}[l_i], M_{\rm label}[s]],
$
where $M_{\rm label}$ is the embedding matrix for labels and $s$ represents the label index. This encoding method generates trajectories considering the label information.
With this generated labeled data, we extract features following the method proposed by Liao et al.~\cite{liao2007learning} and train a random forest classifier using these features. 
The lower part of Table~\ref{tab:classification_table} shows the classification results on the test dataset.
The Random Forest results represent the model trained with raw data. Compared to the baseline, our model performs better, but there is still a significant performance drop. This is because trajectory classification requires maintaining spatio-temporal correlations, making it more challenging than next location prediction. Additionally, our model uses discretized time information with the embedding matrix, and the preprocessing required for embedding inherently leads to some information loss.
}
\subsubsection{\textcolor{black}{Commuting distance analysis}}
\textcolor{black}{
Analyzing commuting patterns is a real-world application, as conducted by the U.S. Census Bureau~\cite{machanavajjhala2008privacy}. In this study, similar to the work by Machanavajjhala~\cite{machanavajjhala2008privacy}, we conducted a commuting distance analysis on the generated data. We utilized trajectories labeled as "commuting" from the peopleflow data.
We created $20,000$ trajectories using grid-based discretization with $w=32$ for these commuting trajectories. We then trained a generative model (HRNet or the baseline model) to generate synthetic commuting trajectories, and analyzed them. Figures~\ref{fig:od} show the results.
Each point in the figures corresponds to a specific destination cell. The x-axis shows the number of trajectories ending in that destination cell, and the y-axis shows the average commute distance to that destination block.
For the baseline model (the left figure), many average commute distances are longer than the original average commute distances, especially for destinations with a smaller number of trajectories. In contrast, our generated trajectories (the right figure) result in average commute distances that are similar to the original ones, even for destinations with a smaller number of trajectories.
}

\label{sec:o_d_analysis}
\begin{figure}[t]
    \centering
    \begin{minipage}{.47\columnwidth}
        \centering
        \includegraphics[width=\linewidth]{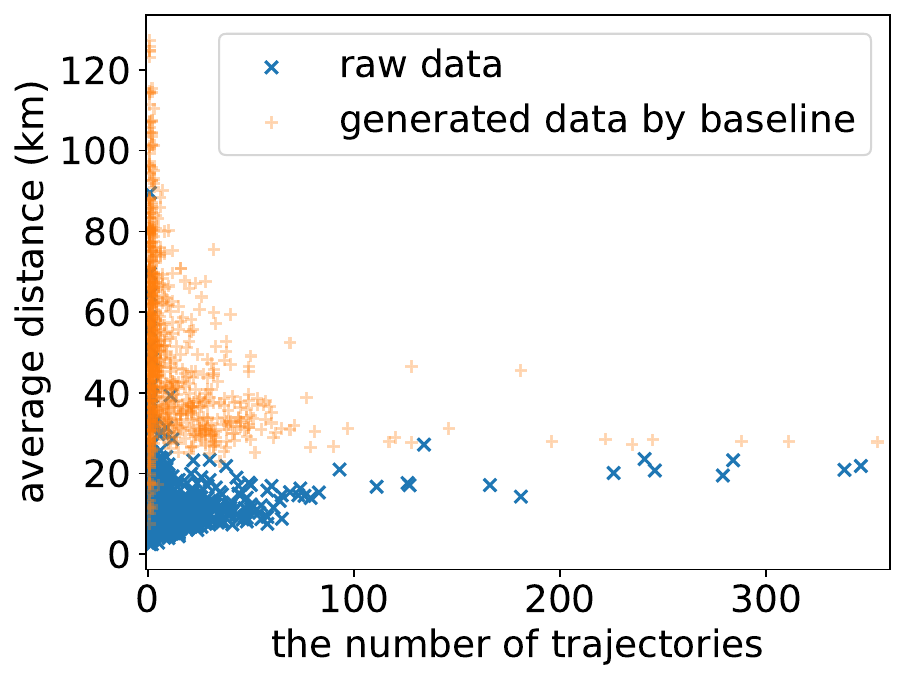}
    \end{minipage}%
    \hspace{10pt}
    \begin{minipage}{.47\columnwidth}
        \centering
        \includegraphics[width=\linewidth]{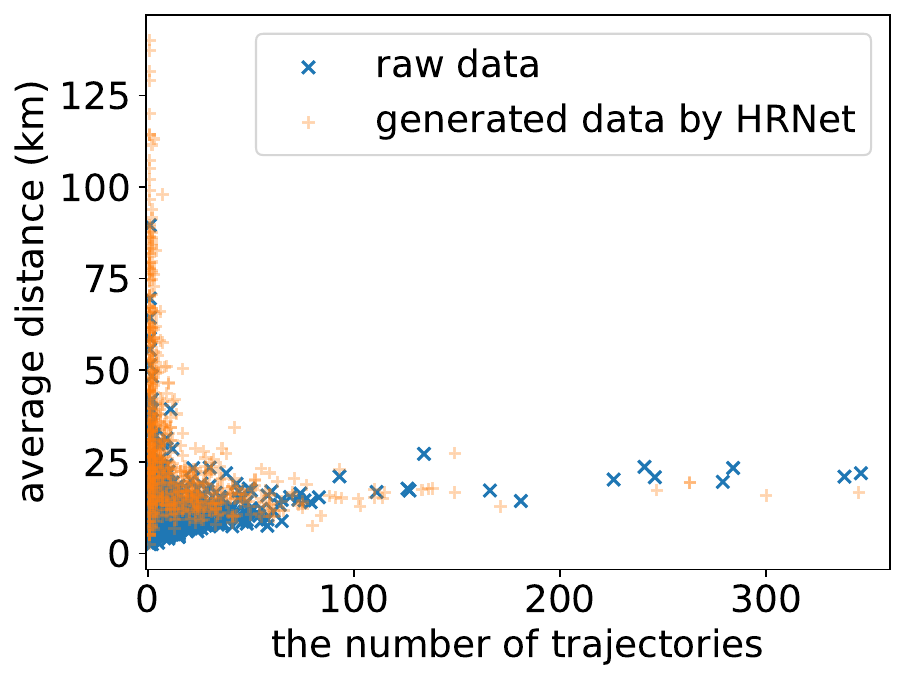}
    \end{minipage}%
    \caption{\textcolor{black}{Average commuting distance for each destination by the baseline model (left) and HRNet (right).}\label{fig:od}}
\end{figure}

\section{Related Work}
\label{sec:related work}
\subsection{Non-deep Learning Models}
\paragraph{Clustering}
Clustering-based methods~\cite{hua2015differentially, li2017achieving} first cluster the locations to reduce the number of locations as a preprocessing step. 
Then, they count the transitions in the smaller location domain (i.e., classes) to create a probabilistic model with the Laplace mechanism.
Many clustering-based methods~\cite{hua2015differentially, li2017achieving} have been found to cause unexpected privacy leakage in the data-dependent preprocessing due to flaws in the privacy proof~\cite{miranda2023sok}. 
Instead, data-dependent preprocessing and the use of the Laplace mechanism, which is very effective with respect to privacy-utility trade-off, can sometimes perform well with small $\varepsilon$ values (see Figure~\ref{fig:comp_destination}). 
However, clustering based preprocessing leads to substantial information loss, and even with large $\varepsilon$ values, accuracy cannot be significantly improved. 
\paragraph{Tree and Markov models}
Tree~\cite{chen2012differentially, he2015dpt} and Markov~\cite{wang2023privtrace, gursoy2018differentially} methods first construct a data structure (e.g., prefix tree) for trajectories. Then, they count elements (e.g., nodes) to create the probabilistic model.
Essentially, tree and Markov methods face two problems due to the high uniqueness of patterns in trajectory:
how to define locations (e.g., determining the grid size $w$) and how to resolve sparsity. 
Since the initial study by Chen et al.~\cite{chen2012differentially}, subsequent studies have attempted to address these issues. 
In the work~\cite{chen2013privacy}, substrings are considered for building an exploration tree based on a Markov assumption, resulting in higher leaf counts and better utility. 
DPT~\cite{he2015dpt} addresses the sparsity problem by restricting movement to adjacent cells using a hierarchical structure, effectively capturing short trajectory features. 
However, this method’s limitation to adjacent cell movement prevents conversion to stay point trajectories, resulting in longer sequence lengths. Given DP's characteristics, performance significantly degrades with longer sequence lengths, making it unsuitable for capturing broad patterns such as daily human movement. 
PrivTrace~\cite{wang2023privtrace} addresses the sparsity problem heuristically, yielding good results on some datasets, but there remains the problem of determining constants used in heuristics. 
Our method addresses them using hierarchical networks and multi-task learning.

\subsection{Deep Learning Models}
As argued in this study, when using DP-SGD, sparsity emerges as a significant issue in deep learning, and thus not many studies propose using DP-SGD.
Ahuja et al.~\cite{ahuja2020differentially} propose using negative sampling and skip-gram to solve this sparsity issue.
Note that such techniques could also be applicable to our study. 
Wang et al.~\cite{wang2022deep} address this problem similarly to DPT by restricting movement to adjacent road network segments. Hence, as demonstrated in Figure~\ref{fig:mtnet_destination}, they effectively capture short trajectory features and perform well with small $\varepsilon$ due to prior road network information. 
However, the same problem as DPT arises, making it unsuitable for capturing broad daily human movement patterns.
Without the formal guarantees like DP, many deep learning approaches \cite{fontana2023gans, ozeki2023balancing, song2023except} rely on obfuscation caused by the generation itself. 
However, without any formal guarantee, it is unclear if privacy is genuinely preserved~\cite{stadler2022synthetic}.
Even without privacy-aware considerations, sparsity remains a problem, and there are studies proposing models similar to ours. 
Lim et al.~\cite{lim2022hierarchical} address sparsity by employing multi-task learning with multiple resolutions, but their naive creation of embedding matrices for each resolution results in a large number of parameters ($O(n^2)$), which is not suitable to DP-SGD. 
Lian et al.~\cite{lian2020geography} adopt a hierarchical location encoding using a quadtree similar to our study, but instead of using a deconvolutional network, they treat quadtree paths as sequences and perform sequence embedding with the attention architecture~\cite{vaswani2017attention}. Consequently, encoding a single location requires processing a sequence length corresponding to the quadtree depth. Longer sequence lengths degrade DP-SGD performance, which is not ideal. Moreover, the application of multi-task learning to such methods is not trivial.
Many deep learning methods for capturing trajectory features have been considered~\cite{feng2018deepmove, luo2021stan, lim2020stp, yang2020location}. 
However, directly encoding the tuple of latitude and longitude is fundamentally challenging~\cite{lian2020geography}, leading to the use of embedding matrices. 
Using embedding matrices poses the same issues as the baseline.

\section{Conclusion}
In this paper, we introduced HRNet, a novel framework designed to effectively learn and generate human mobility data under DP constraints. 
HRNet incorporates three novel components: a hierarchical location encoding component, multi-resolution multi-task learning, and private pre-training. 
The concept of a hierarchical network, as employed in HRNet, has intrinsic value beyond its application in location encoding. 
The hierarchical nature of data is not unique to locations but is also evident in language for example, where words possess semantic layers. 
This observation opens up possibilities for applying our hierarchical approach to improve language models under DP constraints, suggesting a promising avenue for future research in the field of privacy-preserving deep learning.

\begin{acks}
This work was supported by the support of JST SICORP JPMJSC2107, JST CREST JPMJCR21M2, JST PRESTO JPMJPR23P5, JSPS KAKENHI JP22H03595, JP23K24851, JP21K19767, NSF CNS-2125530, CNS-2124104, IIS-2302968, and CDC 1NU38FT000001-01-00.
\end{acks}


\bibliographystyle{ACM-Reference-Format}
\balance
\bibliography{sample}

\clearpage

\appendix

\section{More Details}

\subsection{Derivation of the privacy budget allocation}
\label{sec:privacy budget allocation derivation}
We derive the solution~(\ref{eq:privacy budget allocation solution}) as follows. 
The noise introduced by the Laplace mechanism is
$
O\left(\frac{\log (w^2)}{\varepsilon_2}\right)
$
~\cite{vadhan2017complexity}.
Due to the definition of the transition matrix (see Equation~\ref{eq:tran}), the signal in an element of the DP transition matrix is
$
O\left(\frac{|D|}{4^{i_{\rm res}} w^2}\right).
$
Then, based on the noise and signal, SNR is represented as
$
{\rm SNR}=O\left(\frac{4^{i_{\rm res}}w^2\log (w^2)}{|D|\varepsilon_2}\right).
$
We can see that the asymptotic trend of SNR is determined solely based on the meta-information. 
Hence, setting
$
\varepsilon_2=c\left(\frac{4^{i_{\rm res}}w^2\log (w^2)}{|D|}\right)
$ 
ensures a constant SNR, with empirical results indicating that $c=0.018$ offers a balanced approach. 
While heuristic, experimental evidence suggests that this allocation strategy works well across various settings.

\subsection{Dataset:}
\label{sec:dataset}

\textbf{Geolife (human mobility)~\cite{zheng2009mining}:}
This is a dataset of human mobility in Beijing, China, which was collected by Microsoft Research Asia with GPS devices of individuals. 
Each trajectory includes a set of trips of human represented by latitude and longitude, with random time intervals.
Most trajectories are in Beijing, China, so we select the trajectories that are in the specific range in Beijing.

\textbf{Peopleflow (human mobility)~\cite{sekimoto2011pflow}:}
Originating from Tokyo University's research, the Peopleflow dataset encapsulates human mobility data derived from the Person Trip Survey by the Japan International Cooperation Agency. 
Each trajectory  represents daily human movements in Tokyo during 2008 at one-minute intervals.

\textbf{Didi in Chengdu (road network)~\cite{wang2022deep}:}
This dataset focuses on taxi mobility in Chengdu, China, constrained by road networks. 
Released by Didi, this dataset has been partially processed and made available by Wang et al.~\cite{wang2022deep}. 
It features short-distance taxi trips, each trajectory representing a portion of a longer journey.


\subsection{Evaluation metrics:}
\label{sec:evaluation_metrics}
\textbf{Waypoint:}
A typical trajectory comprises of multiple waypoints of semantic significance (e.g., moving from home to a park and then to a restaurant). 
In this analysis, we define stay points (refer to Section~\ref{sec:preprocess} for details) as waypoints. 
The probability of a location $l\in L$ being a waypoint, given the start POI $l_{\rm start}$, is evaluated using the formula:
\begin{equation}
\label{eq:prob_way}
\Pr(l\in \mathbf{v}| v_0=l_{\rm start}).
\end{equation}

\textbf{Destination:}
The final POI in a trajectory holds significant semantic value as the destination. 
We calculate the probability of a location $l\in L$ being the destination, given the starting POI $l_{\rm start}$, as follows:
\begin{equation}
\label{eq:prob_destination}
\Pr (\mathbf{v}[-1]=l|v_0=l_{\rm start}).
\end{equation}
Here, $\mathbf{v}[-1]$ is the last POI of $\mathbf{v}$. 
In this context, we highlight that this forms a probability distribution over all POIs because $\sum_{l\in L} \Pr (\mathbf{v}[-1]=l|v_0=l_{\rm start})=1$.

\textbf{Transition:}
The transition from a starting location to subsequent points is a key characteristic of a trajectory. 
We assess the probability of transitioning to a location $l\in L$ as the next waypoint from the first POI.
\begin{equation}
\label{eq:prob_transition}
\Pr (l=\mathbf{v}[1]|v_0=l_{\rm start}).
\end{equation}
Similar to destination, this forms a probability distribution over all POIs $L$. 

\textbf{Travel distance:}
The total distance traversed in a trajectory is a vital aspect of human mobility characterization. We evaluate the travel distance range $r\in R$ covered by $\mathbf{v}$ using the formula:
\begin{equation}
\label{eq:prob_travel}
\Pr \left(\sum^{|\mathbf{v}|}_{i=0} d({\mathbf{v}}[i],\mathbf{v}[i+1])\in r\right).
\end{equation}
In this equation, $R$ comprises ranges $[0,d_{\rm max}/n_{\rm bin}),[d_{\rm max}/n_{\rm bin},2*d_{\rm max}/n_{\rm bin})\dots, [(n_{\rm bin}-1)*d_{\rm max}/n_{\rm bin},d_{\rm max}]$ with $d_{\rm max}$ being the maximum distance and the integer $n_{\rm bin}$ representing the bin size. 
The function $d$ denotes the Euclidean distance on the map. This metric forms a probability distribution over all ranges $R$.

\textbf{\textcolor{black}{Diameter:}}
\textcolor{black}{
The diameter indicates the maximum distance between any two points in a trajectory. We evaluate the diameter range $r\in R$ covered by $\mathbf{v}$ using the formula:
\begin{equation}
\label{eq:prob_travel}
\Pr \left(\max_{i,j\in [|\mathbf{v}|]} d({\mathbf{v}}[i],\mathbf{v}[j])\in r\right).
\end{equation}
See above for the definition of $R$.
}

\textbf{Route:}
A route is a sequence of passing points to traverse waypoints in $\mathbf{v}$, denoted as ${\rm route}(\mathbf{v})$. 
For trajectories sampled from the original dataset, the original trajectory before being preprocessed is referenced for the route. 
In contrast, for trajectories from a generator, the route is constructed using the shortest path between the two adjacent waypoints.
Note that this type of post-processing is unnecessary for MTNet, as it directly learns and generates route trajectories, rather than stay-point trajectories. 
The probability of a POI $l\in L$ being part of the route, given the starting POI $l_{\rm start}$, is calculated as:
\begin{equation}
\label{eq:prob_route}
\Pr (l\in {\rm route}(\mathbf{v})|v_0=l_{\rm start}).
\end{equation}

\textcolor{black}{
\textbf{Density at $t$:}
Density at $t$ represents a list of counts of the trajectories whose location at time $t$ is a specific location.
Let $\mathbf{t}$ be the generated time stamps where $\mathbf{t}[i]$ corresponds to time that one visits the location $\mathbf{v}[i]$.
Then, density\_$t$ is defined as follows:
$$
\Pr(\mathbf{v}[i]=l|\mathbf{t}[i]<t\ \&\  \mathbf{t}[i+1]>t).
$$
}

\textbf{\textcolor{black}{Trajectory density:}}
\textcolor{black}{
It measures the number of trajectories passing through a specific area. We first generate 500 random sets $Q$ of POIs by choosing random number of POIs in the map, and count the number of trajectories passing through each area. 
Denote the set of POIs as $q\in Q$, the trajectory density count in $D$ is as follows:
$$
{\rm count}(q,D)=\sum_{\mathbf{v}\in D}\mathbf{1} \left[\exists i \in [|\mathbf{v}|],\mathbf{v}[i]\in q\right]
$$
}

\textbf{\textcolor{black}{Trajectory pattern:}}
\textcolor{black}{
It captures the frequency of transiting from one place to another.
In most downstream tasks, only frequent patterns are considered, so we use the top $200$ frequent patterns $Q$, which are longer than $2$.
Denote each pattern as $q\in Q$, the trajectory pattern count in $D$ is as follows ${\rm count}(q,D)=$:
$$
\sum_{\mathbf{v}\in D}\mathbf{1} [\exists i,j\in [\mathbf{v}],i<j\ \text{s.t.}\ (\mathbf{v}[i], \mathbf{v}[i+1], \ldots, \mathbf{v}[j])= q],
$$
}

\subsubsection{Discrepancy}
In this section, we detail the methodology for computing discrepancies across various metrics using the JS divergence.
In our study, we define $L_{\rm start}$ as the set of candidate locations that serve as potential starting POIs. 
To determine $L_{\rm start}$, we have selected the top $30$ locations with the highest occurrence rates in our dataset as the initial POIs. 
This selection provides a representative sample for our analysis.

\textbf{Waypoint and route:}
For the metrics of waypoint and route, as defined in probability distributions~(\ref{eq:prob_way}) and (\ref{eq:prob_route}), the support is binary ($\{0,1\}$), indicating the presence or absence of a POI $l\in L$ in a (route) trajectory. 
We denote $p_{l,l_{\rm start}}$ and $q_{l, l_{\rm start}}$ as the probability distributions derived from the original and generated datasets, respectively. 
The discrepancy for these metrics is calculated with the averaged JS divergence as shown in the following equation:
$
\frac{1}{|L_{\rm start}|} \sum_{l_{\rm start}\in L_{\rm start}}\frac{1}{|L|}\sum_{l\in L} {\rm JS}(p_{l,l_{\rm start}}||q_{l,l_{\rm start}}).
$

\textbf{Destination, transition, and travel distance:}
For destination and transition metrics, as outlined in probability distributions~(\ref{eq:prob_destination}) and (\ref{eq:prob_transition}), the support encompasses all POIs ($L$). 
Similarly, for the travel distance metric defined in~(\ref{eq:prob_travel}), the support includes all distance ranges ($R$). 
We assign $p_{l_{\rm start}}$ and $q_{l_{\rm start}}$ to represent the probability distributions from the original and generated datasets, respectively. 
The discrepancy for these metrics is computed with the averaged JS divergence as follows:
$
\frac{1}{|L_{\rm start}|} \sum_{l_{\rm start}\in L_{\rm start}} {\rm JS}(p_{l_{\rm start}}||q_{l_{\rm start}}).
$

\textbf{\textcolor{black}{density at $t$:}}
\textcolor{black}{
For density at $t$, the support encompasses all POIs ($L$).
We assign $p_t$ and $q_t$ to represent the probability distributions from the original and generated datasets for time $t$, respectively. 
Then, the discrepancy is as follows:
$
\frac{1}{|T|}\sum_{t\in T}{\rm JS}(p_t||q_t),
$
where $T$ is the set of all discrete times.
}

\textbf{\textcolor{black}{Trajectory density and trajectory pattern}}
\textcolor{black}{
For trajectory density and trajectory pattern, we will obtain the set of counts given by the set of queries $Q$.
Then, we use the average relative error (ARE) to evaluate the similarity between the counts from the real dataset  $D$ and those from the generated dataset $D^\prime$.
$$
{\rm ARE}=\frac{1}{|Q|}\sum_{q\in Q} \frac{\left|{\rm count}(q,D)-{\rm count}(q,D^\prime)\right|}{\max\{{{\rm count}(q,D)},\phi\}},
$$
where $\phi$ is a constant to bound the impact of a query of small real value.
}

\section{Omited experiments}

\begin{figure}[t]
    \centering
    \begin{minipage}{.47\columnwidth}
        \centering
        \includegraphics[width=\linewidth]{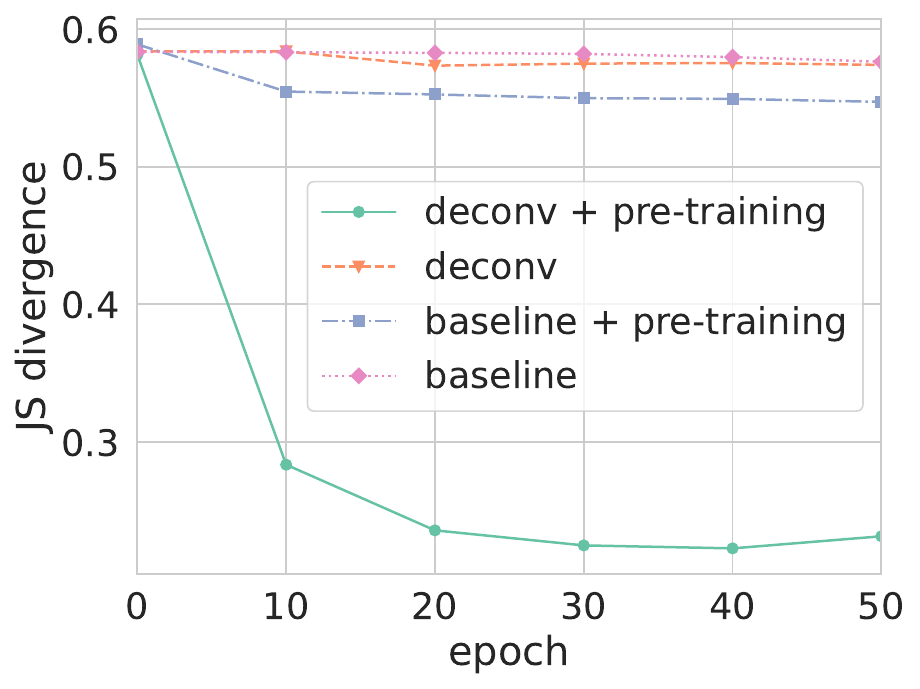}
        \caption{The discrepancy of the destination on the Straight dataset ($w=32$) for each training epoch. \textcolor{black}{deconv means deconvolutional network.}\label{fig:ablation_pre-trainiing}}
    \end{minipage}%
\end{figure}

\subsection{Ablation Study: Pre-training}

Section~\ref{sec:exp_ablation} established that private pre-training significantly enhances performance by helping DP-SGD training. 
However, a question arises: does pre-training merely memorize the DP transition matrix, or does it genuinely aid in DP-SGD training? To address this query, we utilize the Straight dataset for an in-depth analysis.

In our experiment, we employ the first transition ($\Pr(v_1|v_0)$) as the DP transition matrix during the pre-training phase. 
Subsequently, we evaluate the distribution of the final destination ($\Pr(l_2|v_0)$), which remains untrained in the pre-training phase due to the unique structure of the Straight dataset. 
The results of this evaluation are presented in Figure~\ref{fig:ablation_pre-trainiing}.

Our findings indicate that pre-training enhances utility in both scenarios: with only pre-training and when combined with the deconvolutional network. 
This improvement suggests that pre-training contributes more than just memorizing the provided DP transition matrix; it actively supports the training process of DP-SGD. 
Notably, pre-training in conjunction with the deconvolutional network yields more significant benefits compared to pre-training the baseline. 
The reason for this is that while the baseline’s pre-training does not include the location encoding component, the deconvolutional network benefits from pre-training with coarse location encoding, thereby enhancing its overall effectiveness.

\subsection{Privacy Budget Allocation}

\begin{figure}[t]
    \centering
    \begin{minipage}{.47\columnwidth}
        \centering
        \includegraphics[width=\linewidth]{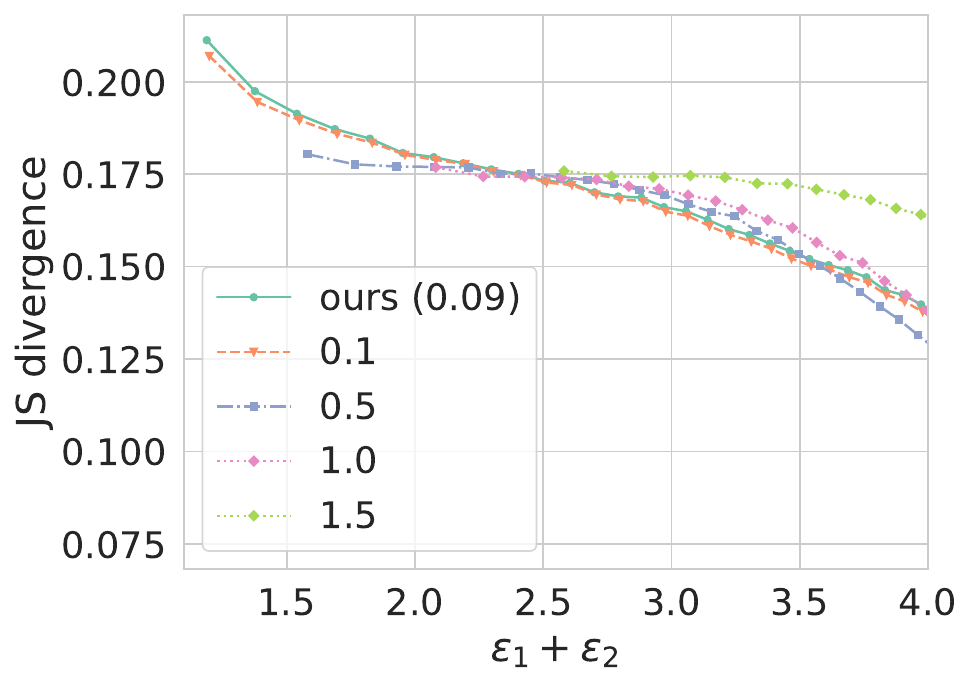}
        \caption*{$w=16$}
    \end{minipage}
    \begin{minipage}{.47\columnwidth}
        \centering
        \includegraphics[width=\linewidth]{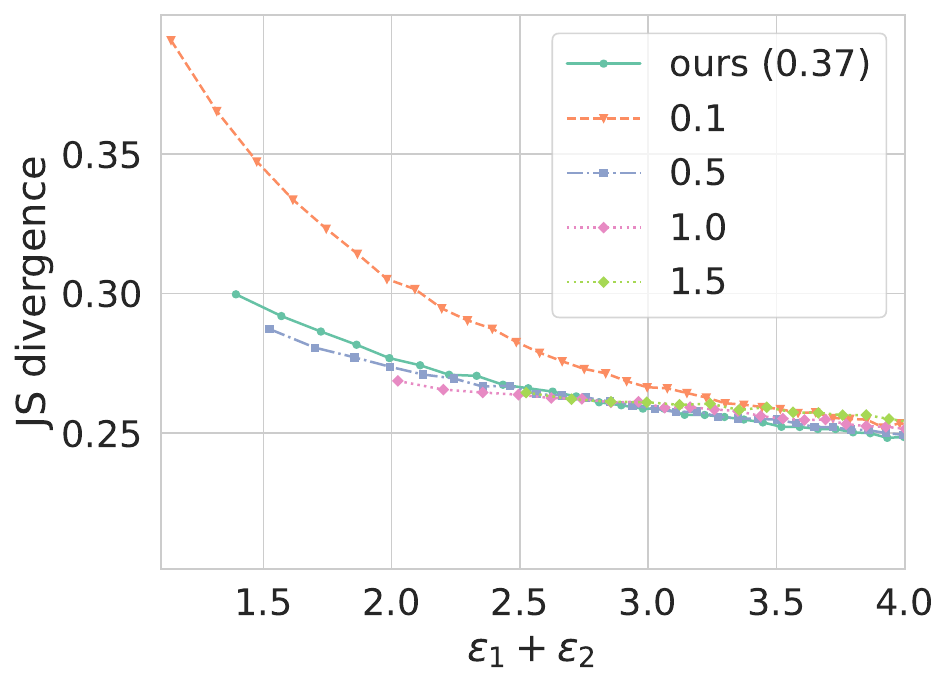}
        \caption*{$w=32$}
    \end{minipage}%
    \begin{minipage}{.47\columnwidth}
        \centering
        \includegraphics[width=\linewidth]{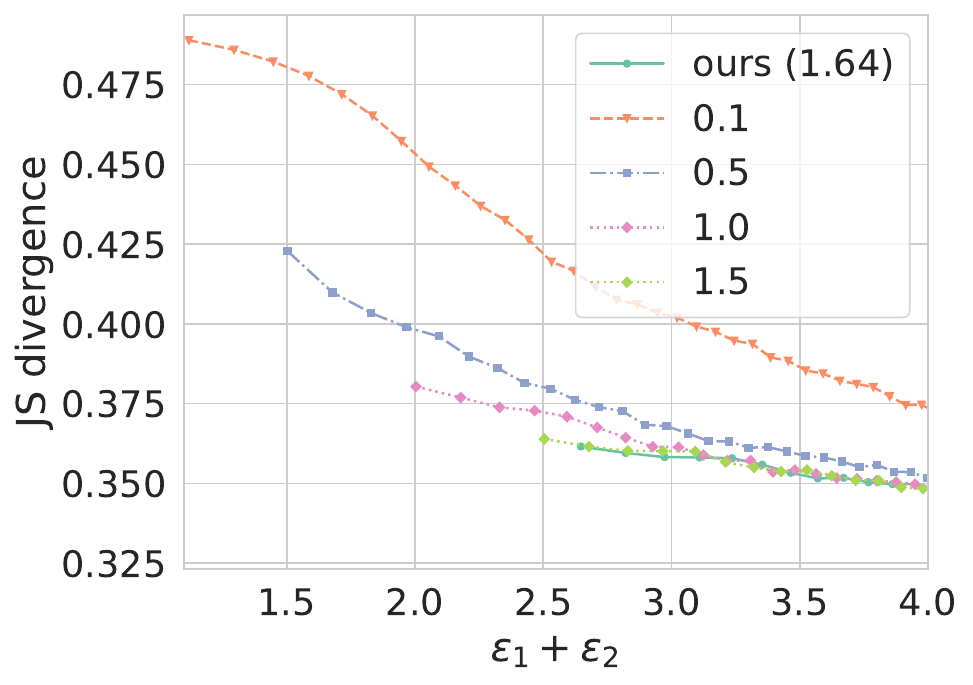}
        \caption*{$w=64$}
    \end{minipage}%
    \caption{\label{fig:budget allocation} Discrepancy in transition on the Peopleflow dataset under three different settings of $w$. Each line illustrates a scenario with a distinct privacy budget used in private pre-training ($\varepsilon_2=0.1, 0.5, 1.0, 1.5$, and the value derived from our allocation method).}
\end{figure}

Determining the allocation of privacy budget between $\varepsilon_1$ and $\varepsilon_2$ is a unique hyperparameter in HRNet. 
In this section, we present experiments on the Peopleflow dataset, demonstrating the impact of privacy budget allocation on performance, as well as the efficacy of our heuristic method proposed in Section~\ref{sec:priv_budget_allocation}. 
As shown in Figure~\ref{fig:budget allocation}, we observe variations in outcomes when altering the budget allocated to private pre-training (i.e., $\varepsilon_2$).

In scenarios where $w$ is small (indicating a lower number of POIs), a small privacy budget $\varepsilon_2$ tends to be more effective. 
This is attributed to the denser nature of the transition matrix in such settings, which does not necessitate a large privacy budget, thereby allowing more budget to be allocated to DP-SGD (i.e., $\varepsilon_1$). 
An example of this can be seen where assigning $\varepsilon_2$ a value of $1.5$ in the $w=16$ setting actually deteriorates results, due to an excessive allocation for the dense transition matrix.

Conversely, in settings where $w$ is large, a large budget for $\varepsilon_2$ is beneficial.
This is because the transition matrix in such instances necessitates a large budget to ensure an adequately high signal-to-noise ratio. 
For instance, setting $\varepsilon_2$ to $0.1$ in the $w=64$ scenario leads to worse results, as the transition matrix becomes too noisy, diminishing its informational value for pre-training.

Our heuristic method adapts automatically to the number of POIs and the number of records, achieving a balanced allocation without seeing the sensitive dataset. 
While this is a heuristic approach, it does not require additional privacy budget to determine its value, making it an efficient and effective strategy for privacy budget allocation in HRNet.



\end{document}